%% file: main.tex
\documentclass[letterpaper, 10 pt, conference]{ieeeconf}  
\usepackage{cite}
\usepackage{amsmath,amssymb,amsfonts}
\usepackage{algorithmic}
\usepackage{graphicx}
\usepackage{textcomp}
\usepackage{mathtools}
\usepackage{caption}
\usepackage{subcaption}
\usepackage{cleveref}
\usepackage{etoolbox}

\usepackage{todonotes}

\usepackage{stackengine}

\newtheorem{definition}{Definition}
\newtheorem{theorem}{Theorem}
\newtheorem{remark}{Remark}

\newtheorem{lemma}{Lemma}

\newtheorem{corollary}{Corollary}
\newtheorem{example}{Example}

\crefname{section}{Section}{Sections}
\crefname{subsection}{Section}{Sections}
\crefname{definition}{Definition}{Definitions}
\crefname{proposition}{Proposition}{Propositions}
\crefname{lemma}{Lemma}{Lemmas}
\crefname{theorem}{Theorem}{Theorems}
\crefname{corollary}{Corollary}{Corollaries}
\crefname{example}{Example}{Examples}
\crefname{assumption}{Assumption}{Assumptions}
\crefname{notation}{Notation}{Notations}
\crefname{remark}{Remark}{Remarks}
\crefname{running}{Running Example}{Running Examples}
\crefname{algorithm}{Algorithm}{Algorithms}
\def\delequal{\mathrel{\ensurestackMath{\stackon[1pt]{=}{\scriptstyle\Delta}}}}
\DeclareMathAlphabet{\mathpzc}{OT1}{pzc}{m}{it}
\makeatletter
\newcommand{\algmargin}{\the\ALG@thistlm}
\makeatletter
\newcommand*{\rom}[1]{\expandafter\@slowromancap\romannumeral #1@}
\makeatother
\def\BibTeX{{\rm B\kern-.05em{\sc i\kern-.025em b}\kern-.08em
    T\kern-.1667em\lower.7ex\hbox{E}\kern-.125emX}}
\begin{document}
\title{Vulnerability Analysis of Nonlinear Control Systems to Stealthy False Data Injection Attacks}
\author{Amir Khazraei, and Miroslav Pajic
\thanks{This work is sponsored in part by the ONR agreements N00014-17-1-2504 and N00014-20-1-2745, AFOSR award FA9550-19-1-0169, and the NSF CNS-1652544 award, as well as the National AI Institute for Edge Computing Leveraging Next Generation Wireless Networks, Grant CNS-2112562. Some of preliminary results in this
paper appeared in~\cite{khazraei2022resiliencycdc}.}
\thanks{The authors are with the Department of Electrical and Computer
Engineering, Duke University, Durham, NC 27708 USA (e-mail:
amir.khazraei@duke.edu , miroslav.pajic@duke.edu). }}

\maketitle

\begin{abstract}
In this work, we focus on analyzing vulnerability of nonlinear dynamical control systems to stealthy false-data injection attacks on sensors. We start by defining the \emph{stealthiness} notion 
in the most general form 
where an attack is considered stealthy if it would be undetected by  \emph{any} 
intrusion detector -- i.e., 
any intrusion detector 
could not do better than a random  guess. Depending on the level of attacker's knowledge about the plant model, controller, and the system states, two different attack models are considered. For each attack model, we derive the conditions for which the system will be vulnerable to  stealthy impactful attacks, in addition to finding a methodology for designing such sequence of false-data injection attacks. 
When the attacker has complete knowledge about the system, 
we show that  if the closed-loop system is incrementally exponentially stable while the open-loop plant is incrementally unstable, then the 
system is vulnerable to stealthy yet impactful attacks on sensors. However, in the second attack model, 
with less knowledge about the system, 
additional conditions need to be satisfied and the level of stealthiness depends on the accuracy of attacker's knowledge about the system. 
We also 
consider the impact of stealthy attacks on state estimation, and show  that if the closed-loop control system including the estimator is incrementally stable, then the state estimation 
in the presence of attack converges to the attack-free estimates. Finally, we illustrate our results on numerical case-studies.

\end{abstract}


\input{Intro}

\input{Preliminaries}
\input{Motive}

\input{Perfect}
\input{estimation}

\input{Simulation}

\input{Conclusion}

\bibliographystyle{IEEEtranMod}

\input{appendix}

\end{document}

%% file: Intro.tex
\section{Introduction}  
\label{sec:intro}

Control systems have been shown to be vulnerable to a wide range of cyber and physical attacks with severe consequences. (e.g.,~\cite{chen2011lessons}). As part of the control design and analysis,~it~is~thus critical to identify early any vulnerability of the considered system to impactful attacks, especially the ones that are potentially stealthy to the deployed intrusion detection mechanisms.

Depending on the resources available to the attacker, different types of stealthy impactful attacks have been proposed. For instance, for LTI systems with strictly proper transfer functions, by  compromising the control input, the attacker  can design effective  stealthy attacks if the system has unstable zero invariant; e.g.,~\cite{teixeira2012revealing} where such attack is referred to as the zero dynamics attack. However, when the transfer function is not strictly proper, the attacker needs to compromise both plant's inputs and outputs. When the attacker compromises both the plant's actuation and sensing, 
e.g., ~\cite{sui2020vulnerability} derives the conditions under which the system is vulnerable to stealthy~attacks. Other types of attacks that targets both input and output also exist for LTI systems including replay attack~\cite{mo2009secure,mo2013detecting} and covert attack~\cite{smith2015covert}. Specifically, the authors in~\cite{mo2013detecting} show that replay attacks can bypass $\chi^2$ intrusion detector (ID) for some class of LQG controllers in LTI systems and remain stealthy. 

On the other hand, sensor attacks, commonly referred to as false-data injection attacks, have drawn a great deal of attention. For example, vulnerability to false data injection attacks, of static state estimation in systems such as power grids was considered in~\cite{liu2011false}; 
by adding values that lie in the range of the observation matrix, 
such attack can bypass $\chi^2$ detectors and lead to incorrect state estimation. However, 
for dynamical systems, merely inserting constant values in the range of the observation matrix would not be stealthy and effective; 
for stealthiness, a false-data injection attack needs to  evolve with some dynamics to remain stealthy~\cite{mo2010false,jovanov_tac19,kwon2014analysis,khazraei2022attack,khazraei_acc20,zhang2020false}.

Specifically, for linear time-invariant (LTI) systems with Gaussian noise, 
if measurements from all sensors can be compromised, the plant's (i.e., open-loop) instability 
is a sufficient condition for an attacker being able to significantly impact the system while remaining undetected (i.e., stealthy) by a particular type of residual-based ($\chi^2$) IDs~\cite{mo2010false,jovanov_tac19,kwon2014analysis,zhang2020false}. These works also show that to construct such attack sequences, the attacker needs to have complete knowledge about the system, including the dynamical model of the plant as well as the controller and Kalman filter gains. Yet, for LTI systems with bounded noise, the plant's instability is a necessary and sufficient condition for the existence of impactful, stealthy attacks when all senors are compromised~\cite{khazraei2022attack,khazraei_acc20}.   
 
All these results~\cite{teixeira2012revealing,sui2020vulnerability,mo2009secure,mo2013detecting,smith2015covert,mo2010false,jovanov_tac19,kwon2014analysis,khazraei2022attack,khazraei_acc20,zhang2020false,shang2021optimal} 
only address LTI systems. Moreover, the notion of stealthiness is only characterized  for a \emph{specific type} of the employed ID (e.g., $\chi^2$-based detectors or RSE detectors). 
The notion of attack stealthiness independent of the employed ID (i.e., remaining stealthy for \emph{all} existing/potential IDs) for LTI systems is studied in~\cite{bai2017data,bai2017kalman,shang2021optimal,fang2019stealthy,zhang2019optimal,shang2021worst}. One of the main differences of this work is that  
%
%
our notion of stealthiness, initially introduced in~\cite{khazraei2022attacks}, is stronger than the one from~\cite{bai2017data,bai2017kalman,shang2021optimal,fang2019stealthy,zhang2019optimal,shang2021worst} where stealthiness depends on time;  i.e., stealthiness  from~\cite{bai2017data,bai2017kalman,shang2021optimal,fang2019stealthy,zhang2019optimal,shang2021worst} requires that only for a bounded time 
the  attack is guaranteed to stay undetected by 
any ID. However, the notion of stealthiness in our work is independent of time and the attack is guaranteed to be stealthy for all time steps after initiating the attack.  Moreover, the performance degradation metric used in~\cite{bai2017data} is the error covariance of a Kalman filter estimator as opposed in our work; we assume the attacker's goal is to cause deviation in the trajectories of the states.



To the best of our knowledge, no existing work provides 
\emph{vulnerability analysis for systems with nonlinear dynamics, while considering general control and ID designs}, as well as provide generative models for such stealthy and impactful attacks. 
In~\cite{smith2015covert}, covert attacks are introduced as stealthy attacks that can target a potentially nonlinear system. 
However, the attacker needs to have perfect knowledge of the system's dynamics and be able  to compromise \emph{both} the plant's input and outputs. Even more importantly, as the attack design is based on attacks on LTI systems, no guarantees are provided for effectiveness and stealthiness of attacks on nonlinear systems. In~\cite{sasahara2022attack}, attack design for nonlinear systems with an arbitrary but fixed ID  (the type of ID is known to the attacker) 
is framed as an optimization problem in a discrete Markov decision system; however, 
no analysis is provided to determine which classes of systems are vulnerable to stealthy impactful attacks. Moreover, the design is based on discretizing the continuous state space where the complexity of the problem increases for higher dimensions. On the other hand, we show that if the condition on the existence of impactful stealthy attacks holds, our generated attack sequence can bypass any ID and the attacker does not need to know what type of ID is deployed. 

More recently, \cite{zhang2021stealthy} introduced stealthy attacks on a \emph{specific class} of nonlinear systems with residual-based IDs, but provided effective attacks only when \emph{both} plant's inputs and outputs are compromised by the attacker. On the other hand, in this work, we assume the attacker can only compromise the  plant's sensing data and  consider systems with \emph{general} nonlinear dynamics. For systems with general nonlinear dynamics and residual-based IDs, machine learning-based methods to design the stealthy attacks have been introduced (e.g.,~\cite{khazraei2021learning}), but without any theoretical analysis and guarantees regarding the impact of the stealthy attacks.

\subsection{Paper Contribution and Organization}
To the best of our knowledge, this is the first work that considers 
existence and design of \emph{impactful} sensor attacks on systems with general nonlinear dynamics such that the attacks are also \emph{stealthy} 
for \emph{any} deployed ID.
The main contributions of this paper are summarized as follows.

First, we introduce two different attack models depending on the level of system knowledge the attacker has. For each attack model, we provide conditions 
that a nonlinear system is vulnerable to effective yet stealthy attacks 
without limiting the analysis to any particular type of employed IDs. Specifically, for the first attack model with a higher level of the system knowledge, we show that if the closed-loop control system is incrementally exponentially stable while the open-loop control system is incrementally unstable, then the system is vulnerable to impactful attacks that can remain stealthy from any ID. For the second attack model that requires less knowledge about the system, additional conditions are imposed on the attacker; we show that if the closed-loop system is incrementally input to state stable and the open-loop system is incrementally input to state unstable, then the system is vulnerable to impactful attacks and the level of stealthiness depends on the accuracy of attacker’s knowledge about the system.  We also show that for LTI systems, if a certain subset of sensors are under attack, the closed-loop system is asymptotically stable, and the open-loop system is unstable, then the system is vulnerable to stealthy attacks independent of the deployed ID; this is a generalization of results from~\cite{mo2010false,jovanov_tac19,kwon2014analysis,zhang2020false} showing that such condition is sufficient for the stealthiness 
under only one class of IDs ($\chi^2$).

Second, for each 
attack model, we provide a general method to a sequence of stealthy and impactful attack vectors. We show that as the dynamical model of the system becomes `simpler' (e.g., moving from highly nonlinear to LTI), the attacker  needs lower levels of system knowledge to generate the attack sequence. In an extreme case, 
for LTI systems, we show that the attacker only needs to have access to the the state transition and the observation matrices, unlike~\cite{mo2010false,jovanov_tac19,kwon2014analysis,zhang2020false} that assume the attacker has access to the full plant model information as well as the controller and Kalman filter gain. 

Finally, we consider the impact of the proposed stealthy attacks on the nonlinear state estimators. We show that if the closed-loop control system, which includes the dynamics of the plant and the estimator, is incrementally exponentially stable, then the state estimation in the presence of attack converges exponentially to the attack free estimates.

The paper is organized as follows. In~\cref{sec:prelim}, we introduce preliminaries, whereas \cref{sec:motive} presents the system and attack model, before formalizing the notion of stealthiness in \cref{sec:stealthy}. 
\cref{sec:perfect} provides sufficient conditions for existence of the impactful yet stealthy attacks. Section~\ref{sec:estimation} provides the impact of such attacks on state estimation. 
Finally, in~\cref{sec:simulation}, we illustrate our results on two case-studies, before concluding remarks in \cref{sec:concl}. 

\subsubsection*{Notation}
We use $\mathbb{R, Z}, \mathbb{Z}_{\geq 0}$ to denote  the sets of reals, integers and non-negative integers, respectively, and $\mathbb{P}$ denotes the probability for a random variable.  For 
a square matrix $A$, $\lambda_{max}(A)$ denotes the maximum eigenvalue and if $A$ is symmetric, then $A\succ 0$ denoted the matrix is positive definite. 
For a vector $x\in{\mathbb{R}^n}$, $||x||_p$ denotes the $p$-norm of $x$; when $p$ is not specified, the 2-norm is implied. 
For a vector sequence, 
$x_0:x_t$ denotes the set $\{x_0,x_1,...,x_t\}$. 
A function $f:\mathbb{R}^{n}\to \mathbb{R}^{p}$ is Lipschitz with constant $L$ if for any $x,y\in \mathbb{R}^{n}$ it holds that $||f(x)-f(y)||\leq L ||x-y||$. 
%
Finally, if $\mathbf{P}$ and $\mathbf{Q}$ are probability distributions relative to Lebesgue measure with densities $\mathbf{p}$ and $\mathbf{q}$, respectively, then 
the Kullback–Leibler  (KL) divergence between $\mathbf{P}$ and $\mathbf{Q}$ is defined as
$KL(\mathbf{P},\mathbf{Q})=\int \mathbf{p}(x)\log{\frac{\mathbf{p}(x)}{\mathbf{q}(x)}}dx$.

%% file: Preliminaries.tex
\section{Preliminaries}\label{sec:prelim}

Let $\mathbb{X}\subseteq \mathbb{R}^n$ and $\mathbb{D}\subseteq \mathbb{R}^m$. Consider a discrete-time nonlinear system with an exogenous input, modeled as 
\begin{equation}\label{eq:prilim}
x_{t+1}=f(x_t,d_t),\quad x_t\in \mathbb{X},\,\,t\in \mathbb{Z}_{\geq 0},
\end{equation}
where $f:\mathbb{X}\times\mathbb{D}\to \mathbb{X}$ is continuous. We denote by $x(t,\xi,d)$ 
the trajectory (i.e., the solution) of~\eqref{eq:prilim} at time $t$, when the system has the initial condition $\xi$ and is subject to the input sequence $d=\{d_0:d_{t-1}\}$.

The following definitions are derived from~\cite{angeli2002lyapunov,tran2018convergence,tran2016incremental}.

\begin{definition}\label{def:IES}
The system~\eqref{eq:prilim} is incrementally exponentially stable (IES) in the set $\mathbb{X}\subseteq \mathbb{R}^n$ if exist $\kappa>1$ and
$\lambda<1$ that 
\begin{equation}
\Vert x(t,\xi_1,d)-x(t,\xi_2,d)\Vert \leq \kappa \Vert \xi_1-\xi_2\Vert \lambda^{t},
\end{equation}
holds for all $\xi_1,\xi_2\in \mathbb{X}$, any $d_t\in \mathbb{D}$, and $t\in \mathbb{Z}_{\geq 0}$. When $\mathbb{X}=\mathbb{R}^n$, the system is referred to as globally incrementally exponentially stable (GIES).
\end{definition}

\begin{definition}
The system~\eqref{eq:prilim} is incrementally input to state stable (IISS) in the set $\mathbb{X}\subseteq \mathbb{R}^n$ if there exists a function $\gamma\in \mathcal{K}_{\infty}$, $\kappa>1$, and
$\lambda<1$ such that for any $t\in \mathbb{Z}_{\geq 0}$ 
\begin{equation}\label{eq:IISS}
\Vert x(t,\xi_1,d^1)-x(t,\xi_2,d^2)\Vert \leq \kappa \Vert \xi_1-\xi_2\Vert \lambda^{t}+\gamma(\Vert d^1-d^2\Vert_{\infty}),
\end{equation}
holds for all $\xi_1,\xi_2\in \mathbb{X}$, any $d_t^1,d_t^2\in \mathbb{D}$, and $t\in \mathbb{Z}_{\geq 0}$, where $\Vert d^1-d^2\Vert_{\infty}=\sup_{t\geq 0}{\Vert d_t^1-d_t^2\Vert_{\infty}}$. When $\mathbb{X}=\mathbb{R}^n$, the system is referred to as globally incrementally input to state stable~(GIISS).
\end{definition}

By replacing $d^1=d^2$ in~\eqref{eq:IISS}, one can verify that if the system is IISS, then it is also IES. 
\begin{definition}
The system~\eqref{eq:prilim} is incrementally  unstable (IU) in the set $\mathbb{X}\subseteq \mathbb{R}^n$ if for all $\xi_1\in \mathbb{X}$ and any $d_t\in \mathbb{D}$, there exists an $\xi_2$ arbitrarily close to $\xi_1$ such that for any $M>0$, 
\begin{equation}\label{eq:IU}
\Vert x(t,\xi_1,d)-x(t,\xi_2,d)\Vert \geq M,
\end{equation}
holds for all $t\geq t'$, for some $t'\in \mathbb{Z}_{\geq 0}$.
\end{definition}

\begin{definition}
The system~\eqref{eq:prilim} is incrementally input to state unstable (IISU) in the set $\mathbb{X}\subseteq \mathbb{R}^n$ if for all $\xi_1\in \mathbb{X}$ and any $d_t^1,d_t^2\in \mathbb{D}$ satisfying $\Vert d_t^1-d_t^2\Vert \leq \delta$ for all $t\in \mathbb{Z}_{\geq 0}$, there exists an $\xi_2$ such that for any $M>0$, 
\begin{equation}\label{eq:IISU}
\Vert x(t,\xi_1,d^1)-x(t,\xi_2,d^2)\Vert \geq M,
\end{equation}
holds for all $t\geq t'$, for some $t'\in \mathbb{Z}_{\geq 0}$.
\end{definition}

Similarly, by replacing $d^1=d^2$ in~\eqref{eq:IISS}, one can verify that if the system is IISU, then it is also IU. 

Now, we present some properties of Kullback–Leibler (KL) divergence known as \emph{monotonicity} and \emph{chain-rule}~\cite{polyanskiy2022information}.


\begin{lemma}{\cite{polyanskiy2022information}}\label{lemma:mon}
\textbf{(Monotonicity):} Let $P_{X,Y}$ and $Q_{X,Y}$ be two distributions for a pair of variables $X$ and $Y$, and $P_{X}$ and $Q_{X}$ be two distributions for variable $X$. Then, 
\begin{equation}
KL(Q_X||P_X)\leq KL(Q_{X,Y}||P_{X,Y})
\end{equation}
\end{lemma}

\begin{lemma}{\cite{polyanskiy2022information}}\label{lemma:chain}
\textbf{(Chain rule):} Let $P_{X,Y}$ and $Q_{X,Y}$ be two distributions for a pair of variables $X$ and $Y$. Then,
\begin{equation}
KL(Q_{X,Y}||P_{X,Y})= KL(Q_{X}||P_{X})+KL(Q_{Y|X}||P_{Y|X}),
\end{equation}
\end{lemma}
where $KL(Q_{Y|X}||P_{Y|X})$ is defined as
\begin{equation}
KL(Q_{Y|X}||P_{Y|X})=\mathbb{E}_{x\sim Q_X}\{KL(Q_{Y|X=x}||P_{Y|X=x})\}.
\end{equation}

\begin{lemma}{\cite{polyanskiy2022information}}\label{lemma:Guassian}
Let $P_{X}$ and $Q_{X}$ be two Gaussian distributions with the same covariance $\Sigma$ and different means of $\mu_P$ and $\mu_Q$, respectively. Then, it holds that 
\begin{equation}
KL(Q_{X}||P_{X})= \mu_Q^T\Sigma^{-1} \mu_P.
\end{equation}
\end{lemma}

\begin{lemma}\label{lemma:maximum}
Let $Q_{X}$ be a distribution for a random variable $X$ and we have $X\leq M$ for some $M>0$. Then, 
\begin{equation}
\mathbb{E}_{Q_X}\{X\}\leq M
\end{equation}
\end{lemma}
\begin{proof}
The proof directly follows from the definition of expectation and some properties of integral.  
\end{proof}

%% file: Motive.tex
\section{System and Attack Model} 
\label{sec:motive} 
In this section, we introduce the considered system and attack model, allowing us to formally capture the problem addressed in this work.  We consider the setup from Fig.~\ref{fig:architecture} where each of the components is modeled as follows. 

\subsubsection{Plant}
We assume that the states of the system evolve following a general nonlinear discrete-time dynamics that can be captured in the state-space form~as  
\begin{equation}\label{eq:plant}
\begin{split}
{x}_{t+1} &= f(x_t,u_t)+w_t,\\
y_t &= h(x_t)+v_t;
\end{split}
\end{equation}
here, $x \in {\mathbb{R}^n}$, $u \in {\mathbb{R}^m}$, $y \in {\mathbb{R}^p}$ are the state, input and output vectors of the plant, respectively. 
We assume that $h$ is Lipschitz with a constant $L_h$.  The plant output vector captures measurements from the set of plant sensors $\mathcal{S}$. 
Further, $w \in {\mathbb{R}^{n}}$ and $v \in {\mathbb{R}^p}$ are the process and measurement noises that  assumed to be i.i.d Gaussians with zero mean, and $\mathbf{R}_w \succ 0$ and $\mathbf{R}_v \succ 0 $ covariance matrices, respectively. We also assume the system starts operating at time $-T$ with initial state $x_{-T}$.

As we show later, it will be useful to consider the input to state relation of the dynamics~\eqref{eq:plant}; if we define $U=\begin{bmatrix}u^T&w^T\end{bmatrix}^T$, the first equation in~\eqref{eq:plant} becomes
\begin{equation}\label{eq:input-state}
x_{t+1}=f_u(x_t,U_t).
\end{equation}

\begin{figure}[!t]
\centering
\includegraphics[width=0.468\textwidth]{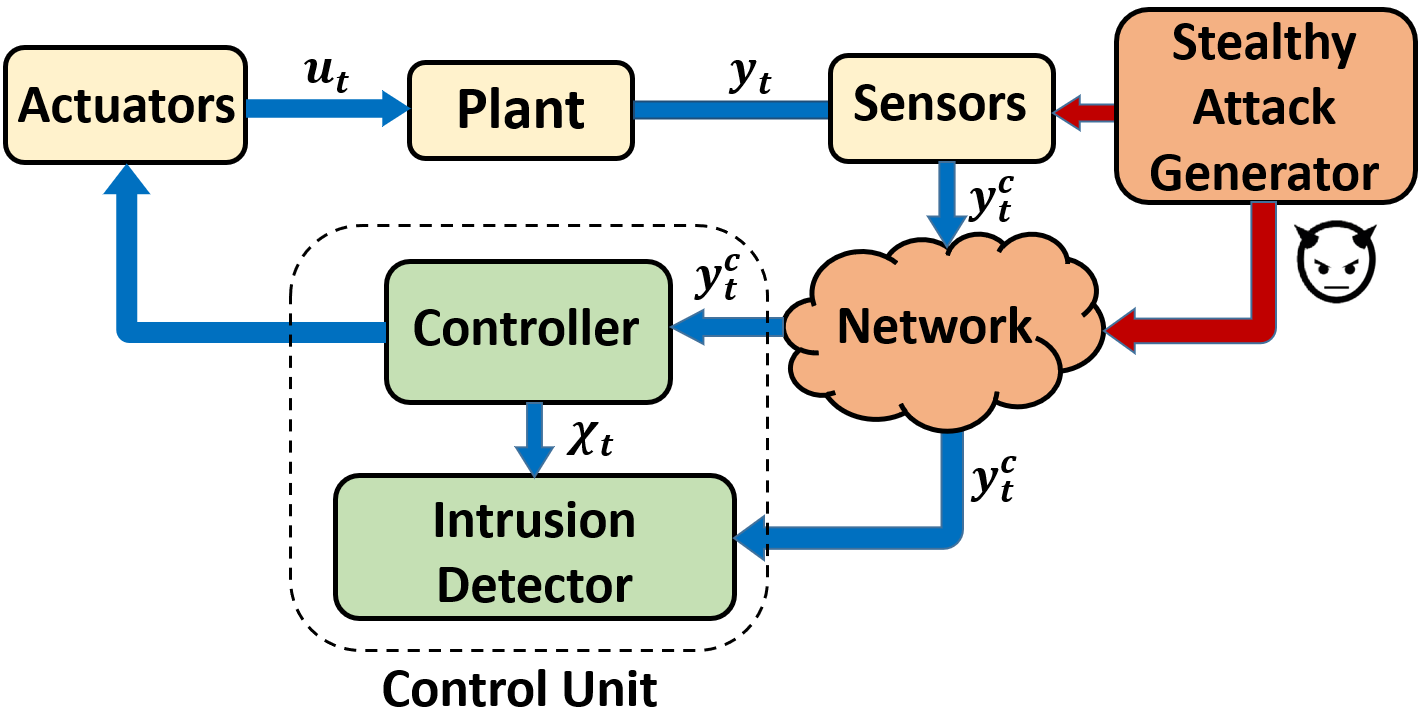}
\caption{Control system architecture considered in this work, in the presence of network-based attacks on sensing data.}
\label{fig:architecture}
\end{figure}

\subsubsection{Control Unit}
The controller, illustrated in Fig.~\ref{fig:architecture}, is equipped with a feedback controller in the most general form, as well as an intrusion detector (ID). In what follows, we provide more details on the controller design. Intrusion detector will be discussed after introducing the attack model.

\paragraph*{Controller}
Due to their robustness to uncertainties, closed-loop controllers are utilized in 
most control systems. In the most general form, a feedback controller can be captured in the state-space form~as 
\begin{equation}
\label{eq:control}
\begin{split}
\mathpzc{X}_{t} &= f_c(\mathpzc{X}_{t-1},y_t^c),\\
u_t  &= h_c(\mathpzc{X}_{t},y_t^c),
\end{split}
\end{equation}
where $\mathpzc{X}$ is the internal state of the controller, and $y^c$ captures the sensor measurements received by the controller. 
%
Thus, without malicious activity, it holds that $y^c=y$.\footnote{We assume that the communication network is reliable (e.g., wired).}
Note that the control model~\eqref{eq:control} is general, capturing for instance 
nonlinear filtering followed by a classic nonlinear controller (e.g., $f_c$ can model an extended Kalman filter and $h_c$ any full-state feedback controller). {Here, we assume that the function $f_c$ is Lipschitz with a constant $L_{f_c}$}.
We also assume $\mathpzc{X}_{-T}$ is obtained deterministically by the system operator.\footnote{The result of the chapter holds even when such initial condition is chosen randomly. We discuss more about this in the proof of Theorem~\ref{thm:PAt}.}

The dynamics of the closed-loop system can be captured~as
\begin{equation}\label{eq:closed-loop}
\mathbf{X}_{t+1}=F(\mathbf{X}_t,\mathbf{W}_t).
\end{equation}
with the full state of the closed-loop system  $\mathbf{X}\delequal \begin{bmatrix}
{x}_t^T&{\mathpzc{X}}_t^T \end{bmatrix}^T$, and exogenous disturbances $\mathbf{W}_t\delequal \begin{bmatrix}
w_t^T&{v}_{t+1}^T&{v}_{t}^T \end{bmatrix}^T$. Therefore, the functions $f_c$ and $h_c$ are designed such that the closed-loop system states satisfy some desired properties. 


\subsubsection{Attack Model} \label{sec:attack_model}

We consider a sensor attack model where, for  sensors from a set $\mathcal{K}\subseteq{\mathcal{S}}$, 
the information delivered to the controller differs from the non-compromised sensor measurements. The attacker can achieve this via e.g., noninvasive attacks such sensor spoofing (e.g.,~\cite{kerns2014unmanned}) or by compromising information-flow from the sensors in $\mathcal{K}$ to the controller (e.g., as in network-based attacks~\cite{lesi_rtss17,lesi_tcps20}). In either cases,  the attacker can launch false-date injection attacks, inserting a desired value instead of the current measurement of a compromised sensor.\footnote{We refer to sensors from $\mathcal{K}$ as compromised, even if a sensor itself is not directly
compromised but its measurements may be altered due to e.g., network-based attacks.}

Thus, assuming that the attack starts at time $t=0$, the sensor measurements delivered to the controller for $t\in \mathbb{Z}_{\geq 0}$ can be modeled~as
\begin{equation}\label{att:model}
y^{c,a}_t = y_t^a+a_t;
\end{equation}
here, $a_t\in {\mathbb{R}^p}$ denotes the attack signal injected by the attacker at time $t$ via the compromised sensors from $\mathcal{K}$, $y_t^a$ is the true sensing information (i.e., before the attack is injected at time~$t$) -- we use the
superscript a to differentiate all signals of the attacked system.
In the rest of the paper, we assume $\mathcal{K}=\mathcal{S}$; for some systems, we will discuss how the results can be generalized for the case when $\mathcal{K}\subset \mathcal{S}$. 
Now, with different levels of runtime knowledge about the plant and its states we consider two different attack models as follows.

\textbf{ Attack model~${\text{\rom{1}}}$:} The attacker has perfect knowledge about the system states and the control input, as well as the functions $f$ and $h$  in~\eqref{eq:plant}. 

\textbf{ Attack model~${\text{\rom{2}}}$:} The attacker has only some imperfect knowledge about the system states. The attacker imperfectly reconstructs the states by using either the system's sensor measurements or its own set of external sensors. Moreover, the attacker has knowledge about the controller functions $f_c$ and $h_c$ in~\eqref{eq:control} and the functions $f$ and $h$  in~\eqref{eq:plant} but does not have access to control input $u_t$.  

Consider the system with plant~\eqref{eq:plant}, controller~\eqref{eq:control} and an ID that we define in the next subsection. We denote such system with \emph{attack model~${\text{\rom{1}}}$} and \emph{attack model~${\text{\rom{2}}}$}
as $\Sigma_{\text{\rom{1}}}$ and  $\Sigma_{\text{\rom{2}}}$, respectively.
Note that since the controller uses the received sensing information to compute the input $u_t$, the compromised sensor values affect the evolution of the system and controller states. Hence, we add the superscript $a$ to denote any signal obtained from a compromised system 
-- e.g.,  thus, $y_t^a$ is used to denote before-attack sensor measurements when the system is under attack in~\eqref{att:model}, and we denote the closed-loop plant and controller state when the system is compromised as $\mathbf{X}^a\delequal \begin{bmatrix}
{x^a}\\{{\mathpzc{X}^a}} \end{bmatrix}$. Since the attack starts at time zero, it takes one time step to affect the states and actual output; therefore, we have $y_0^a=y_0$ and $x_0^a=x_0$.

As discussed in the attack models, in this work, we consider the commonly adopted threat model as in majority of existing stealthy attack designs, e.g.,~\cite{mo2009secure,mo2010false,smith2015covert, khazraei2022attack, jovanov_tac19}, where the attacker has full knowledge of the system, its dynamics and employed architecture. In addition, the attacker has the required computational power to calculate suitable attack signals to be injected,
while planning ahead as needed. 


Finally, the \emph{\textbf{attack goal}} is to design an attack signal $a_t$, $t\in \mathbb{Z}_{\geq 0}$, such that it always remains \emph{stealthy} -- i.e., undetected by an employed ID -- while \emph{maximizing control performance degradation}. The notions of \emph{stealthiness} and \emph{control performance degradation} depend on the employed control architecture, and thus will be formally defined after both the controller and ID have been~introduced.


\subsubsection{Intrusion Detector (ID)} 
To detect system attacks (and anomalies), we assume that an 
ID is employed, analyzing the received sensor measurements. Specifically, by defining $Y\delequal {y^c}$, 
as well as $Y^a\delequal{y^{c,a}}$ when the system is under attack,
we assume that the ID has access to a sequence of values $Y_{-T}:Y_t$ (all observations since the system starts operating until time $t$)\footnote{It should be noted that $-T$ in this notation does not have its mathematical meaning. Here, we used it to show the initial time of the system.} 
and solves the binary hypothesis 
checking\\

\vspace{-4pt}
$H_0$:  normal condition (the ID receives $Y_{-T}:Y_t$);~~

$H_1$: abnormal behaviour (the ID receives 
$Y_{-T}^{-1},Y_{0}^a:Y_t^a$),\footnote{Since the attack starts at $t=0$, we do not use superscript $a$ for the system evolution for $t<0$, as the trajectories of the non-compromised and compromised systems do not differ before the attack~starts.}\vspace{-4pt}
\\

\noindent
where we denote $Y_{-T}^{-1}=Y_{-T}:Y_{-1}$. Given a sequence of received data denoted by $\bar{Y}^t=\bar{Y}_{-T}:\bar{Y}_t$, it is either extracted from the $H_0$ hypothesis with the joint distribution denoted as $\mathbf{P}(Y_{-T}:Y_t)$, or from the alternate hypothesis with a joint (\emph{unknown}) distribution  denoted by $\mathbf{Q}(Y_{-T}^{-1},Y_{0}^a:Y_t^a)$;\footnote{With some abuse of the notation, we just use $\mathbf{P}$ or $\mathbf{Q}$ to refer to these joint distributions.}  
{note that the joint distribution $\mathbf{Q}$ is controlled by the injected attack signal and is thus unknown.} 

We define the intrusion detector $D$ as the mapping
\begin{equation}
D: \bar{Y}^t\to \{0,1\},
\end{equation}
where the output $0$ is associated with $H_0$ and output 1 with $H_1$. It should be noted that ID in the above format is general as there is no assumption on the mapping $D$, and the output measurements are the only information that is accessible to the system to detect the abnormalities. Any other signal including $\mathpzc{X}$ (consequently the input control $u_t$) can be obtained using the sequence of sensor measurements $\bar{Y}^t$ where such mapping can be captured in mapping $D$. 

Let us define $p_t^{TD}(D)=\mathbb{P}(D(\bar{Y}^t)=1|\bar{Y}^t \sim \mathbf{Q})$ as the probability of true detection, and $p_t^{FA}(D)=\mathbb{P}(D(\bar{Y}^t)=1|\bar{Y}^t \sim \mathbf{P})$ as the probability of false alarm for the detector $D$. Let also consider the random guess-based ID (defined by $D_{RG}$). For such random guess-based ID we have 
\begin{equation*}
\begin{split}
p^{FA}(D_{RG}) \overset{\text{\small (a)}}{=}& \mathbb{P}(D_{RG}(\bar{Y}^t)=1|\bar{Y}^t \sim \mathbf{P}) 
\overset{\text{\small (b)}}{=}\mathbb{P}(D_{RG}(\bar{Y}^t)=1)\\
\overset{\text{\small (c)}}{=}& \mathbb{P}(D_{RG}(\bar{Y}^t)=1|\bar{Y}^t \sim \mathbf{Q})\overset{\text{\small (d)}}{=}p^{TD}(D_{RG}),
\end{split}
\end{equation*}
where (a) and (d) hold according to the definition of true detection and false alarm probabilities; (b) and (c) hold because in random guess the detection is independent of the distribution of the observation. 
Hence, for random guess detectors, the probability of true detection and false alarm are equal. 

\begin{definition}\label{def:random_guess}
An ID (defined by $D$) is better than a random guess ID (defined by $D_{RG}$) if $p^{FA}(D)< p^{TD}(D)$.
\end{definition}

\begin{remark}
Note that although we assumed the ID has access to the measurement data from time $-T$ to the current time $t$, it does not mean that the system only should individually process each observation at each time step. The ID can combine all the observation until current time step to decide if the system is in normal condition. 
\end{remark} 

We now formalize the notion of stealthy attacks. 

\section{Formalizing Stealthiness and Attack Objectives} 
\label{sec:stealthy}
This section captures the conditions that an attack~sequence is stealthy 
from \emph{any} ID. Specifically, we define an attack to be strictly stealthy if \emph{there exists no detector that can perform better than a random guess between the two hypothesis} (Definition~\ref{def:random_guess}). However, reaching such stealthiness guarantees may not be possible in general. Thus, we also define the notion of $\epsilon$-\emph{stealthiness}, which as we will show later, is attainable for a large class of nonlinear systems. Formally, we define the notions of \emph{strict stealthiness} and \emph{$\epsilon$-stealthiness} as~follows.

\begin{definition}
\label{def:stealthiness}
Consider the system from~\eqref{eq:plant}. An attack sequence, 
denoted by $\{a_{0}, a_{1},...\}$, 
is 
\emph{\textbf{strictly stealthy}} if there exists no detector for which $p_t^{TD}-p_t^{FA}>0$ holds, for any $t\geq 0$. 
An attack is
\textbf{$\epsilon$-\emph{stealthy}} if for a given $\epsilon >0$, there exists no detector such that $p_t^{TD}-p_t^{FA}>\epsilon$ holds, for any $t\geq 0$. 
\end{definition}





The following theorem uses Neyman-Pearson lemma to capture the condition for which the received sensor measurements 
satisfy the stealthiness condition in Definition~\ref{def:stealthiness}. 


\begin{theorem}[\cite{khazraei2022attacks,khazraei_l4dc22}]\label{thm:stealthy}
An attack sequence 
is 
\begin{itemize}
    \item 
    strictly stealthy if and only if  
    $KL\big(\mathbf{Q}(Y_{-T}^{-1},Y_{0}^a:Y_t^a)||\mathbf{P}(Y_{-T}:Y_t)\big)=0$ for all $t\in \mathbb{Z}_{\geq 0}$. 
    \item is $\epsilon$-stealthy if the  observation sequence $Y_{0}^a:Y_t^a$ satisfies
    \begin{equation*}
        KL\big(\mathbf{Q}(Y_{-T}^{-1},Y_{0}^a:Y_t^a)||\mathbf{P}(Y_{-T}:Y_t)\big)\leq \log(\frac{1}{1-\epsilon^2}).
    \end{equation*}
\end{itemize}
\end{theorem}

\begin{remark}
The $\epsilon$-stealthiness condition from~\cite{bai2017kalman} requires $$\lim_{t\to \infty}\frac{KL\big(\mathbf{Q}(Y_{0}^a:Y_t^a)||\mathbf{P}(Y_{0}:Y_t)\big)}{t}\leq \epsilon.$$ 
This allows for the KL divergence to linearly increase over time for any $\epsilon>0$, and as a result, after large-enough time period the attack will be detected. On the other hand, our definition of $\epsilon$-stealthy 
only depends on $\epsilon$ and is fixed for any time $t$; thus, it introduces a stronger notion of attack~stealthiness. 
\end{remark}
  
\subsubsection*{Formalizing Attack Goal}\label{sec:attack_goal}

The attacker intends to
\emph{maximize} control performance degradation. Specifically, the attack goal is to cause deviation in the system's  trajectory. In other words, if we assume the attack starts at $t=0$, and denote the states of the attack-free and under attack systems as $x_t$ and $x_t^a$, respectively for $t\in \mathbb{Z}_{\geq 0}$, then 
the attack objective is to achieve   \begin{equation}
\Vert x_{t'}^a-x_{t'}\Vert \geq \alpha.
\end{equation}
for some $t'\in \mathbb{Z}_{\geq 0}$. In other words, the attacker wants to cause deviation in the trajectory of states with respect to system's own desired unattacked trajectory. Moreover, the attacker wants \emph{to remain stealthy (i.e., undetected by the intrusion detector)},
as formalized below.

\begin{definition}
\label{def:eps_alpha}
An attack sequence 
is referred to as $(\epsilon,\alpha)$-\emph{successful} attack if exists $t'\in \mathbb{Z}_{\geq 0}$ such that $ \Vert x_{t'}^a -x_{t'} \Vert \geq \alpha$ and the 
attack is $\epsilon$-stealthy 
for all $t\in \mathbb{Z}_{\geq 0}$.
When such sequence exists for a system, the system is called $(\epsilon,\alpha)$-\emph{attackable}. 
When the system 
is $(\epsilon,\alpha)$-attackable for arbitrarily large $\alpha$, the system is referred to as \emph{perfectly attackable}.
\end{definition}


Now, the problem considered in this work can be formalized as 
capturing the potential impact of stealthy attacks on the considered system. Specifically, in the next section,
we derive conditions for 
existence of a \emph{stealthy} yet \emph{effective} attack sequence
$a_{0}, a_{1},...$  resulting in
$\Vert x_t^a-x_t\Vert \geq \alpha$ for some $t\in \mathbb{Z}_{\geq 0}$ -- i.e., we find conditions for the system to be $(\epsilon,\alpha)$-attackable. 
Here, for an attack to be stealthy, we focus on 
the $\epsilon-$stealthy notion; 
i.e., that even the best 
ID could only improve the detection probability by $\epsilon$ compared to the random-guess baseline detector.

%% file: Perfect.tex
\section{Vulnerability Analysis of Nonlinear Systems to Stealthy Attacks} \label{sec:perfect}

In this section, we derive the conditions such that the nonlinear system~\eqref{eq:plant} with closed-loop dynamics~\eqref{eq:closed-loop} is vulnerable to effective stealthy attacks formally defined in Section~\ref{sec:stealthy}. 

\subsection{Vulnerability Analysis of Nonlinear Systems $\Sigma_{\text{\rom{1}}}$}

First, we derive the condition such that the system $\Sigma_{\text{\rom{1}}}$ is vulnerable to stealthy attacks. 

\begin{theorem}
\label{thm:PAt}
The system~$\Sigma_{\text{\rom{1}}}$ is ($\epsilon,\alpha$)-attackable 
for arbitrarily small $\epsilon$ and arbitrarily large $\alpha$, if the closed-loop system~\eqref{eq:closed-loop} is 
IES and the system~\eqref{eq:input-state} is 
IU. 
\end{theorem}

\begin{proof}
Assume that the trajectory of the system and controller states for $t\in \mathbb{Z}_{<0}$ is denoted by {$\mathbf{X}_{-T}:\mathbf{X}_{-1}$}.
Following the attack start at $t=0$, let us consider 
the evolution of the system with and without attacks during $t\in \mathbb{Z}_{\geq 0}$. 
For the system under attack, starting at time zero, the  trajectory {$\mathbf{X}_{0}^a:\mathbf{X}_{t}^a$} 
of the system and controller states is governed by
\begin{equation}\label{eq:attack_trajec}
\begin{split}
x_{t+1}^a=&f(x_t^a,u_t^a)+ w_t^a,\quad y_t^{c,a}=h(x_t^a)+v_t^a+a_t\\
\mathpzc{X}_{t}^a=&f_c(\mathpzc{X}^a_{t-1},y_t^{c,a}),\quad u_t^a=h_c(\mathpzc{X}^a_{t},y_t^{c,a}).\\
\end{split}
\end{equation}

On the other hand, if the system was not under attack 
during $t\in \mathbb{Z}_{\geq 0}$,  we denote the plant and controller state evolution by {$\mathbf{X}_{0}:\mathbf{X}_{t}$}. Hence, it is a continuation of the system trajectories {$\mathbf{X}_{-T}:\mathbf{X}_{-1}$} if hypothetically no data-injection attack occurs during $t\in \mathbb{Z}_{\geq 0}$. Since the system and measurement noises are independent of the state, we can assume that $w_t^a=w_t$ and $v_t^{a}=v_t$. In this case, the dynamics of the plant and controller state evolution satisfies
\begin{equation}\label{eq:free_trajec}
\begin{split}
x_{t+1}=&f(x_t,u_t)+ w_t,\quad y_t^c=h(x_t)+v_t,\\
\mathpzc{X}_{t}=&f_c(\mathpzc{X}_{t-1},y_t^c),\quad u_t=h_c(\mathpzc{X}_{t},y_t^c),\\
\end{split}
\end{equation}
captured in the compact form~\eqref{eq:closed-loop}, 
with 
$\mathbf{X}_{0}=\begin{bmatrix}x_{0}\\\mathpzc{X}_{0}\end{bmatrix}$. 

Now, consider the sequence of attack vectors injected in the system~\eqref{eq:attack_trajec},  constructed by the attacker using the dynamics
\begin{equation}\label{eq:attack_seq}
\begin{split}
s_{t+1}&=f(x_t^a,u_t^a) -   f(x_t^a-s_t,u_t^a) \\
a_t&=h(x_t^a-s_t)-h(x_t^a),
\end{split}
\end{equation}
for $t\in \mathbb{Z}_{\geq 0}$, and with some arbitrarily chosen nonzero initial value of $s_0$. By injecting the above attack sequence into the sensor measurements, we can verify that $y_t^{c,a}=h(x_t^a)+v_t+a_t=h(x_t^a-s_t)+v_t$. After defining\footnote{Superscript $f$ is the initial of the word `fake' as we will show later how the attacker fool the system to believe $x^f$ represents as the systems state.}
\begin{equation}\label{eq:fake_state}
x_t^f\delequal x_t^a-s_t,
\end{equation}
and combining~\eqref{eq:attack_seq} with~\eqref{eq:attack_trajec}, the dynamics of $x_t^f$ and the controller, as well as the corresponding input and output satisfy
\begin{equation}\label{eq:closed_loop_attack}
\begin{split}
x_{t+1}^f=&f(x_t^f,u_t^a)+ w_t,\quad y_t^{c,a}=h(x_t^f)+v_t,\\
\mathpzc{X}_{t}^a=&f_c(\mathpzc{X}^a_{t-1},y_t^{c,a}), \quad \,\,\,\,\,\,\, u_t^a=h_c(\mathpzc{X}^a_{t},y_t^{c,a}),
\end{split}
\end{equation}
with the initial condition $x_0^f=x_0^a-s_0$. 

Now, if we define $\mathbf{X}^f_t=\begin{bmatrix}x_t^f\\\mathpzc{X}_{t}^a\end{bmatrix}$, it holds that
\begin{equation}\label{eq:closed_attack}
\mathbf{X}^f_{t+1}=F(\mathbf{X}^f_{t},\mathbf{W}_t).
\end{equation}
with $\mathbf{X}^f_{0}=\begin{bmatrix}x_0^f\\\mathpzc{X}_{0}^a\end{bmatrix}$. Since we have that $x_0^a=x_0$, it holds that $\mathbf{X}_{0}-\mathbf{X}^f_{0}=\begin{bmatrix}s_{0}\\\mathpzc{X}_{0}-\mathpzc{X}_{0}^a \end{bmatrix}$. Since the functions $f_c$ and $h$ are Lipschitz with constants $L_{f_c}$ and $L_h$, respectively, we have 
\begin{equation*}
\begin{split}
\Vert \mathpzc{X}_{0}-\mathpzc{X}_{0}^a \Vert &\leq \Vert f_c(\mathpzc{X}_{-1},y^{c}_0)-f_c(\mathpzc{X}_{-1},y^{c,a}_0) \Vert \leq L_{f_c} \Vert y^{c}_0 - y^{c,a}_0 \Vert\\
&\leq L_{f_c}\Vert h(x_0)+v_0-h(x_0-s_0)-v_0\Vert \\
&\leq L_{f_c} L_h \Vert s_0\Vert.
\end{split}
\end{equation*}
Therefore, we get $\Vert \mathbf{X}_{0}-\mathbf{X}^f_{0} \Vert \leq (1+L_{f_c} L_h)\Vert s_0\Vert$.

On the other hand, since  both~\eqref{eq:closed_attack} and~\eqref{eq:closed-loop} have the same function and argument $\mathbf{W}_t$, 
the closed-loop system~\eqref{eq:closed_attack} is 
IES, and it also follows that
\begin{equation}\label{eq:error_bound}
\begin{split}
\Vert \mathbf{X}(t,\mathbf{X}_{0},\mathbf{W})-\mathbf{X}^f(t,\mathbf{X}^f_{0},\mathbf{W})\Vert &\leq \kappa \Vert \mathbf{X}_{0}-\mathbf{X}^f_{0}\Vert \lambda^{t}\\
&= \kappa (1+L_{f_c} L_h) \Vert s_0\Vert \lambda^{t},
\end{split}
\end{equation}
for some nonnegative $\lambda <1$. Therefore, the trajectories of $\mathbf{X}$ (i.e., the system without attack)  and $\mathbf{X}^f$ converge to each other exponentially fast. 

Now, by defining $\mathbf{Z}_{t}=\begin{bmatrix}x_{t}\\y_t^{c}\end{bmatrix}$ and $\mathbf{Z}_{t}^f=\begin{bmatrix}x_t^f\\y_t^{c,a}\end{bmatrix}$, it holds that
\begin{equation}\label{ineq:1_2}
\begin{split}
KL\big(&\mathbf{Q}(Y_{-T}^{-1},Y_{0}^a:Y_t^a)||\mathbf{P}(Y_{-T}:Y_t)\big)  \\
&\leq  KL\big(\mathbf{Q}(\mathbf{Z}_{-T}^{-1},\mathbf{Z}_{0}^f:\mathbf{Z}_{t}^f)||\mathbf{P}(\mathbf{Z}_{-T}:\mathbf{Z}_{t})\big),
\end{split}
\end{equation}
where we applied the monotonicity property of the KL-divergence from Lemma~\ref{lemma:mon} to get the above inequality.  Then, we apply the chain-rule property of KL-divergence on the right-hand side of 
\eqref{ineq:1_2} to obtain the following
\begin{equation}\label{ineq:2}
\begin{split}
&KL\big(\mathbf{Q}(\mathbf{Z}_{-T}^{-1},\mathbf{Z}_{0}^f:\mathbf{Z}_{t}^f)||\mathbf{P}(\mathbf{Z}_{-T}^{-1},\mathbf{Z}_{0}:\mathbf{Z}_{t})\big)\\
&= KL\big(\mathbf{Q}(\mathbf{Z}_{-T}^{-1})||\mathbf{P}(\mathbf{Z}_{-T}^{-1})\big)+\\
&\quad \quad \quad \quad KL\big(\mathbf{Q}(\mathbf{Z}_{0}^f:\mathbf{Z}_{t}^f|\mathbf{Z}_{-T}^{-1})||\mathbf{P}(\mathbf{Z}_{0}:\mathbf{Z}_{t}|\mathbf{Z}_{-T}^{-1})\big)\\
&=KL\big(\mathbf{Q}(\mathbf{Z}_{0}^f:\mathbf{Z}_{t}^f|\mathbf{Z}_{-T}^{-1})||\mathbf{P}(\mathbf{Z}_{0}:\mathbf{Z}_{t}|\mathbf{Z}_{-T}^{-1})\big);
\end{split}
\end{equation}
here, we used the fact that the 
KL-divergence of two identical joint distributions (i.e., $\mathbf{Q}(\mathbf{Z}_{-T}^{-1})$ and $\mathbf{P}(\mathbf{Z}_{-T}^{-1})$ since the system is not under attack for $t<0$) 
is zero.

Applying the chain-rule property of 
KL-divergence to~\eqref{ineq:2}
\begin{equation}\label{ineq:3}
\begin{split}
& KL\big(\mathbf{Q}(\mathbf{Z}_{0}^f:\mathbf{Z}_{t}^f|\mathbf{Z}_{-T}^{-1})|| \mathbf{P}(\mathbf{Z}_{0}:\mathbf{Z}_{t}|\mathbf{Z}_{-T}^{-1})\big) \\
&= \sum_{k=0}^{t}\Big\{KL\big(\mathbf{Q}(x_k^f|\mathbf{Z}_{-T}^{-1},\mathbf{Z}_{0}^f:\mathbf{Z}_{k-1}^f)||\mathbf{P}(x_{k}|\mathbf{Z}_{-T}:\mathbf{Z}_{k-1})\big)\\
&+ KL\big(\mathbf{Q}(y_k^{c,a}|x_k^f,\mathbf{Z}_{-T}^{-1},\mathbf{Z}_{0}^f:\mathbf{Z}_{k-1}^f)||\mathbf{P}(y_{k}|x_k,\mathbf{Z}_{-T}:\mathbf{Z}_{t-1})\big)\Big\}.
\end{split}
\end{equation}

Given $\mathbf{Z}_{-T}:\mathbf{Z}_{k-1}$, the distribution of $x_k$ is a Gaussian with mean $f(x_{k-1},u_{k-1})$ and covariance $\mathbf{R}_w$. Similarly given  $\mathbf{Z}_{-T}^{-1},\mathbf{Z}_{0}^f:\mathbf{Z}_{k-1}^f$, and the fact that $s_0$ is a deterministic value,  the distribution of $x_k^f$ is a Gaussian with mean $f(x_{k-1}^f,u_{k-1}^a)$ and covariance $\mathbf{R}_w$. According to~\eqref{eq:free_trajec} and~\eqref{eq:closed_loop_attack}, it holds that $f(x_{k-1},u_{k-1})-f(x_{k-1}^f,u_{k-1}^a)=x_k-x_k^f$. On the other hand, in~\eqref{eq:error_bound}, we showed that $\Vert x_k-x_k^f\Vert \leq \kappa (1+L_{f_c} L_h) \Vert s_0\Vert \lambda^{k}$ holds for $k\in \mathbb{Z}_{\geq 0}$. Therefore, for all $0\leq k\leq t$, it holds that
\begin{align}\label{ineq:4}
KL\big(\mathbf{Q}(x_k^f|&\mathbf{Z}_{-T}^{-1},\mathbf{Z}_{0}^f:\mathbf{Z}_{k-1}^f)||\mathbf{P}(x_{k}|\mathbf{Z}_{-T}:\mathbf{Z}_{k-1})\big) \nonumber\\
\stackrel{(i)}=& \mathbb{E}_{\mathbf{Q}(\mathbf{Z}_{-T}^{-1},\mathbf{Z}_{0}^f:\mathbf{Z}_{k-1}^f)}\{(x_k-x_k^f)^T \mathbf{R}_w^{-1}  (x_k-x_k^f) \}
\nonumber\\\stackrel{(ii)}\leq& \kappa^2 \Vert s_0\Vert^2 \lambda^{2k} (1+L_{f_c} L_h)^2 \lambda_{max}(\mathbf{R}_w^{-1}),
\end{align}
where we used Lemma~\ref{lemma:Guassian} to get equality $(i)$ and using the boundedness of  $(x_k-x_k^f)^T \mathbf{R}_w^{-1}  (x_k-x_k^f)$ followed by applying Lemma~\ref{lemma:maximum}, we get the  inequality $(ii)$, with   $\lambda_{max}(\mathbf{R}_w^{-1})$ being the maximum eigenvalue of the matrix~$\mathbf{R}_w^{-1}$. 

Now, it holds that $\mathbf{Q}(y_k^{c,a}|x_k^f,\mathbf{Z}_{-T}:\mathbf{Z}^f_{k-1})=\mathbf{Q}(y_k^{c,a}|x_k^f)$ and $\mathbf{P}(y_{k}|x_k,\mathbf{Z}_{-T}:\mathbf{Z}_{k-1})=\mathbf{P}(y_{k}|x_k)$; 
also, from~\eqref{eq:free_trajec} and~\eqref{eq:closed_loop_attack} it holds that given $x_k$ and $x_k^f$, $\mathbf{P}(y_{k}|x_k)$ and $\mathbf{Q}(y_k^{c,a}|x_k^f)$ are both Gaussian with mean $h(x_k)$ and $h(x_k^f)$, respectively and covariance $\mathbf{R}_v$. Thus, it follows that
\begin{align}\label{ineq:5}
KL\big(\mathbf{Q}&(y_k^{c,a}|x_k^f)||\mathbf{P}(y_{k}|x_k)\big)\nonumber\\
&=\mathbb{E}_{\mathbf{Q}(x_k^f)}\{\big(h(x_k)-h(x_k^f)\big)^T \mathbf{R}_v^{-1}  \big(h(x_k)-h(x_k^f)\big)\}\nonumber\\
&\leq L_h^2 (x_k-x_k^f)^T \mathbf{R}_v^{-1}  (x_k-x_k^f)\nonumber\\
&\leq L_h^2 (1+L_{f_c} L_h)^2\kappa^2 \Vert s_0\Vert^2 \lambda^{2k} \lambda_{max}(\mathbf{R}_v^{-1}),
\end{align}
where we used again Lemmas~\ref{lemma:chain},~\ref{lemma:Guassian} and~\ref{lemma:maximum} to get the above inequality. 
Combining~\eqref{ineq:1_2}-\eqref{ineq:5} results in 
\begin{equation}\label{eq:epsilon}
\begin{split}
&KL\big(\mathbf{Q}(Y_{-T}^{-1},Y_{0}^a:Y_t^a)||\mathbf{P}(Y_{-T}:Y_t)\big)  \leq \\
&\sum_{k=0}^t \kappa^2 \Vert s_0\Vert^2 \lambda^{2k} \lambda_{max}(\mathbf{R}_w^{-1})+L_h^2 \kappa^2 \Vert s_0\Vert^2 \lambda^{2k} \lambda_{max}(\mathbf{R}_v^{-1})\\
&\leq \frac{\kappa^2 (1+L_{f_c} L_h)^2 \Vert s_0\Vert^2}{1-\lambda^2} \big(\lambda_{max}(\mathbf{R}_w^{-1})+L_h^2 \lambda_{max}(\mathbf{R}_v^{-1})\big)\delequal b_{\epsilon}.
\end{split}
\end{equation}
Finally, with $b_{\epsilon}$ defined as in~\eqref{eq:epsilon} and using Theorem~\ref{thm:stealthy} the attack sequence defined in~\eqref{eq:attack_seq} satisfies the $\epsilon$-stealthiness condition with $\epsilon=\sqrt{1-e^{-b_{\epsilon}}}$. 

We now show that the proposed attack sequence is effective; i.e., there exists $t'\in \mathbb{Z}_{\geq 0}$ such that $\Vert x_{t'}^a-x_{t'}\Vert \geq \alpha$ for arbitrarily large $\alpha$. To achieve this, consider the two dynamics from~\eqref{eq:attack_trajec} and~\eqref{eq:closed_loop_attack} for any $t\in \mathbb{Z}_{\geq 0}$
\begin{equation}
\begin{split}
x_{t+1}^a=&f(x_t^a,u_t^a)+ w_t=f_u(x_t^a,U_t^a)\\
x_{t+1}^f=&f(x_t^f,u_t^a)+ w_t=f_u(x_t^f,U_t^a)
\end{split}
\end{equation}
with $U_t^a=\begin{bmatrix}{u_t^a}^T&w_t^T\end{bmatrix}^T$, for $t\in \mathbb{Z}_{\geq 0}$. Since we assumed that the open-loop system~\eqref{eq:input-state} is 
IU, it holds that for all $x_0^a=x_0$, there exits a nonzero $s_0$ such that for any $M>0$ 
\begin{equation}\label{eq:alpha}
\Vert x^a(t,x_0^a,U^a)-x^f(t,x_0^a-s_0,U^a)\Vert \geq M
\end{equation}
holds in $t\geq t'$, for some $t'\in \mathbb{Z}_{\geq 0}$. 

On the other hand, we showed in~\eqref{eq:error_bound} that $\Vert x(t,x_0,U)-x^f(t,x_0^a-s_0,U^a)\Vert \leq \kappa (1+L_{f_c} L_h) \Vert s_0\Vert \lambda^{t}$. Combining this with~\eqref{eq:alpha} results in
\begin{equation}\label{eq:alpha_x_a}
\begin{split}
&\Vert x^a(t,x_0^a,U^a)-x(t,x_0-s_0,U)\Vert = \\
&\Vert x^a(t,x_0^a,U^a)-x^f(t,x_0^a-s_0,U^a)+x^f(t,x_0^a-s_0,U^a)\\&-x(t,x_0-s_0,U)\Vert \geq \Vert x^a(t,x_0^a,U^a)-x^f(t,x_0^a-s_0,U^a)\Vert \\
& -\Vert x^f(t,x_0^a-s_0,U^a)-x(t,x_0-s_0,U)\Vert \\
 & \qquad \qquad \qquad \qquad \geq M-\kappa (1+L_{f_c} L_h)\Vert s_0\Vert \lambda^{t} \\
&\Rightarrow \Vert x^a(t,x_0^a,U^a)-x(t,x_0-s_0,U)\Vert \\
&\geq M-\kappa (1+L_{f_c} L_h) \Vert s_0\Vert \lambda^{t} \geq  M-\kappa (1+L_{f_c} L_h)\Vert s_0\Vert .
\end{split}
\end{equation}
Since $M$ is a free parameter, we can choose it to satisfy $M>\alpha+\kappa (1+L_{f_c} L_h)\Vert s_0\Vert$, for arbitrarily large $\alpha$. Thus, the system is $(\epsilon,\alpha)$-attackable.
\end{proof}

From~\eqref{eq:closed_loop_attack}, 
the false sensor measurements are generated by the evolution of $x_t^f$. Therefore, intuitively, the attacker wants to deceive the system into believing that $x_t^f$ is the actual state of the system instead of $x_t^a$. Since $x_t^f$ and $x_t$ (i.e., the system state if no attack occurs during $t\in \mathbb{Z}_{\geq 0}$) converge to each other exponentially fast, the idea is that the system almost believes that $x_t$ is the system state (under attack), while the actual state $x_t^a$ becomes arbitrarily large (see Fig.~\ref{fig:fake_state}).  On the other hand,  $s_t=x_t^a-x_t^f$ denotes the deviation between the actual state $x_t^a$ and fake state $x_t^f$, and since we have $\Vert x_t^f-x_t\Vert \leq \kappa (1+L_{f_c} L_h) \Vert s_0\Vert \lambda^{t}$, we can derive $\Vert x_t^a-x_t\Vert \geq \Vert s_t \Vert -\kappa (1+L_{f_c} L_h) \Vert s_0\Vert$ using the same procedure as in~\eqref{eq:alpha_x_a}. Thus, for an $s_0$ with very a small norm, $s_t$ can represent the deviation between the under attack $x_t^a$ and attack-free $x_t$ states. This shows that to have an impactful attack, the dynamics of the $s_t$ needs to be unstable.

\begin{figure}[!t]
\centering
\includegraphics[width=0.268\textwidth]{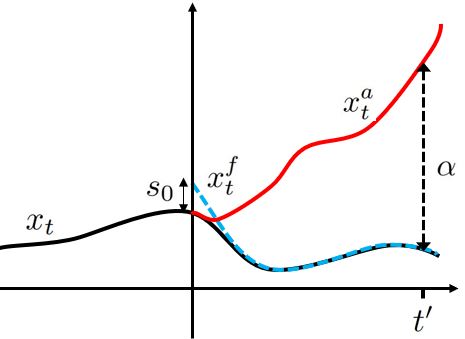}
\caption{The trajectories of the system states $x_t$ without attack (black line), of the `fake 'states $x_t^f$ (blue dash line), and the system states $x_t^a$ after the attack starts (red line).}
\label{fig:fake_state}
\end{figure}

Moreover,
all parameters $\kappa$, $ \lambda$, $L_h$, $L_{f_c}$, $\mathbf{R}_w$, and $\mathbf{R}_v$ in~\eqref{eq:epsilon} are some constants that depend either on system properties ($L_h$, $\mathbf{R}_w$, and $\mathbf{R}_v$) or are determined by the controller design ($L_{f_c}$, $\kappa$, $ \lambda$). However, $s_0$ is set by the attacker, and \emph{\textbf{it can be chosen arbitrarily small} to make $\epsilon$ arbitrarily close to zero}.   Yet, 
$s_0$ cannot be equal to zero; in that case~\eqref{eq:alpha} would not hold --  i.e., the attack would not be impactful. Therefore, as opposed to attack methods targeting the prediction covariance in~\cite{bai2017kalman} where the attack impact linearly changes with $\epsilon$, here arbitrarily large $\alpha$ (high impact attacks) can be achieved  even with an arbitrarily small $\epsilon$; it may only take more time to get to $\Vert x^a_{t'}-x_{t'}\Vert \geq \alpha$.


\subsubsection*{Discussion on Incremental Stability}
The notion of incremental stability from Definition~\ref{def:IES} determines how fast the system `forgets' the impact of the initial condition~\cite{angeli2002lyapunov}.~For~example, consider an LTI state space model $x_{t+1}=Ax_t+Bd_t$. If the matrix $A$ has all eigenvalues inside the unit circle then the system is IES according to Definition~\ref{def:IES}. Alternatively, we can look at the solution of states at time $t$ as 
\begin{equation*}
x_t= A^t x_0 + \sum\nolimits_{k=0}^{t-1} A^{t-k-1}Bd_k,
\end{equation*}
where the term $A^t x_0$ captures the impact of the initial condition and the second term capturing the impact of the input sequence $d_0:d_{t-1}$ on the system.  If all eigenvalues of the matrix $A$ are inside the unit the circle, then the term $A^t x_0$ will approach to zero exponentially for any given initial condition. This means the system will `forget' the impact of initial condition if the matrix $A$ is stable.

Hence, IES means the impact of initial condition converges to zero exponentially fast. Similar notion also exists with an asymptotic behaviour where the system is called incrementally asymptotically stable if the impact of initial condition approaches to zero asymptotically. If we replace the IES condition in Theorem~\ref{thm:PAt} by incrementally asymptotically stability, one can still show that the system will be $(\epsilon,\alpha)$-attackable; yet, the upper bound obtained in~\eqref{eq:epsilon} will be looser. Additional discussion is provided in Appendix~\ref{app:IS}.


\begin{remark}
In case that either $w=0$ or $v=0$ (i.e., when there is no process or measurement noise), one can still get a similar bound on the KL-divergence only as a function of the nonzero noise covariance by applying the monotonicity and data-processing inequalities. However, ensuring stealthiness requirement is not possible if both $w=0$ and $v=0$ (i.e., for a noiseless system), as the system would be completely  deterministic, and thus theoretically any small perturbation to the sensor measurements could be detected.
\end{remark}

\subsection{Vulnerability Analysis of Nonlinear Systems~$\Sigma_{\text{\rom{2}}}$}

As discussed 
in 
Section~\ref{sec:motive}, for the system~$\Sigma_{\text{\rom{2}}}$ 
the attacker does not have perfectly access to the system states. Here, we consider two possible methods for estimating the states.

\textbf{Case 1:} The attacker designs an estimator $\mathcal{E}$ that takes the sequence of system measurements $y_{-\mathcal{T}}:y_{-1},y_0^a:y_{t-1}^a$ to estimate the states
\begin{equation}\label{eq:state_es}
 \hat{x}_t^a=\mathcal{E}_t(y_{-\mathcal{T}}:y_{-1},y_0^a:y_{t-1}^a)
\end{equation}
where $\mathcal{E}$ is a nonlinear mapping that can represent any nonlinear filtering such as extended Kalman filter (EKF). We assume that the attacker uses the measurements of $\mathcal{T}$ (with $-T<-\mathcal{T}$) time step before the attack initiation to estimate the states. It is also adssumed that the estimation error $\zeta=\hat{x}-x$ is bounded; i.e., $\Vert \zeta_t\Vert \leq b_{\zeta}$ for all $t\in \mathbb{Z}_{>0}$.

\textbf{Case 2:} The attacker uses its own set of sensors $y'_t$ to imperfectly measure the states as $y'_t=h'(x_t)+v'_t$ where $v'$ is i.i.d noise independent of $w$ and $v$. Such measurements are used in a filter or fusion model to extract the state estimate~as 
\begin{equation}\label{eq:state_es_2}
 \hat{x}_t^a=\mathcal{E}'_t(y'_{-\mathcal{T}}:y'_{-1},{y'^a_{0}}:{y'^a_{t}}),
\end{equation}
with $\zeta_t=\hat{x}_t^a-x_t^a$. Moreover, it is assumed that $\zeta_t$ is bounded by $b'_{\zeta}$ for all $t\in \mathbb{Z}_{>0}$.

\begin{remark}
Since the system and measurement noises are modeled by Gaussian distribution, theoretically it might not be possible to provide a bound on the estimation error with probability one for case 1. However, we assume $b_{\zeta}$ is chosen high enough that the estimation error is bounded by $b_{\zeta}$ during the attack run-time with probability of almost one. 
\end{remark}

Now, the following theorem 
captures the conditions such that the system~$\Sigma_{\text{\rom{2}}}$ is
vulnerable to stealthy attacks. 

\begin{theorem}\label{thm:PA_2}
{Assume that the function $h_c$ is Lipschitz with a constant $L_{h_c}$}. The system~$\Sigma_{\text{\rom{2}}}$ is ($\epsilon,\alpha$)-attackable 
for arbitrarily large $\alpha$, if the closed-loop system~\eqref{eq:closed-loop} is 
IISS,
the controller dynamics~\eqref{eq:control} is IES, and the system~\eqref{eq:input-state} is 
IISU.
\end{theorem}

\begin{proof}
As in Theorem~\ref{thm:PAt}, let us consider the trajectory of the system and controller states for $t\in \mathbb{Z}_{<0}$ is denoted by {$\mathbf{X}_{-T}:\mathbf{X}_{-1}$}.
Following the attack start at $t=0$, consider 
the evolution of the system with and without attacks during $t\in \mathbb{Z}_{\geq 0}$. 
For the system under attack, starting at $t=0$, the  trajectory {$\mathbf{X}_{0}^a:\mathbf{X}_{t}^a$} 
of the system and controller states satisfies
\begin{equation}\label{eq:attack_trajec_2}
\begin{split}
x_{t+1}^a=&f(x_t^a,u_t^a)+ w_t,\quad y_t^{c,a}=h(x_t^a)+v_t+a_t\\
\mathpzc{X}_{t}^a=&f_c(\mathpzc{X}^a_{t-1},y_t^{c,a}),\quad \quad \, u_t^a=h_c(\mathpzc{X}^a_{t},y_t^{c,a}).\\
\end{split}
\end{equation}

On the other hand, if the system were not under attack 
during $t\in \mathbb{Z}_{\geq 0}$,  we denote the plant and controller state evolution by {$\mathbf{X}_{0}:\mathbf{X}_{t}$}. Hence, it is a continuation of the system trajectories {$\mathbf{X}_{-T}:\mathbf{X}_{-1}$} if hypothetically no data-injection attack occurs during $t\in \mathbb{Z}_{\geq 0}$. Since the system and measurement noises are independent of the state, we can assume that $w_t^a=w_t$ and $v_t^{a}=v_t$. In this case, the dynamics of the plant and controller state evolution satisfies
\begin{equation}\label{eq:free_trajec2}
\begin{split}
x_{t+1}=&f(x_t,u_t)+ w_t,\quad y_t^c=h(x_t)+v_t,\\
\mathpzc{X}_{t}=&f_c(\mathpzc{X}_{t-1},y_t^c),\quad \quad u_t=h_c(\mathpzc{X}_{t},y_t^c),\\
\end{split}
\end{equation}
captured in the compact form~\eqref{eq:closed-loop}, 
with 
$\mathbf{X}_{0}=\begin{bmatrix}x_{0}^T&\mathpzc{X}_{0}^T\end{bmatrix}^T$. 

Now, consider the sequence of attack vectors injected in the system from~\eqref{eq:attack_trajec}, which are constructed by the attacker using the following dynamical model
\begin{subequations}\label{eq:attack_seq_2}
\begin{equation}\label{eq:attack_dynamics}
s_{t+1}=f(\hat{x}_t^a,u_t^s) -   f(\hat{x}_t^a-s_t,u_t^s)
\end{equation}
\begin{equation}\label{eq:attack_vec}
a_t=h(\hat{x}_t^a-s_t)-h(\hat{x}_t^a),
\end{equation}
\begin{equation}\label{eq:control_rec_d}
\mathpzc{X}_{t}^s=f_c(\mathpzc{X}_{t-1}^s,y_t^{c,a}),
\end{equation}
\begin{equation}\label{eq:control_rec}
u_t^s=h_c(\mathpzc{X}_{t}^s,y_t^{c,a})
\end{equation}
\end{subequations}
for $t\in \mathbb{Z}_{\geq 0}$, and with some arbitrarily chosen nonzero initial value of $s_0$ and $\mathpzc{X}_{0}^s$. By injecting the above attack sequence into the sensor measurements, we can verify that
\begin{equation}\label{eq:out_attack}
y_t^{c,a}=h(x_t^a)+v_t+a_t=h(x_t^a)-h(\hat{x}_t^a)+h(\hat{x}_t^a-s_t)+v_t.
\end{equation}
After defining 
\begin{equation}\label{eq:x_f_def}
x_t^f\delequal x_t^a-s_t,
\end{equation}
and combining~\eqref{eq:attack_seq} with~\eqref{eq:attack_trajec}, the dynamics of $x_t^f$ and the controller, and the corresponding input and output satisfy
\begin{subequations}\label{eq:closed_loop_attack_2}
\begin{equation}\label{eq:error_a}
\begin{split}
x_{t+1}^f=&f(x_t^a,u_t^a)-f(\hat{x}_t^a,u_t^s)+f(\hat{x}_t^a-s_t,u_t^s)+w_t\\=&f(x_t^f,u_t^a)-f(x_t^f,u_t^a)+f(x_t^a,u_t^a)-f(\hat{x}_t^a,u_t^s)\\&+f(x_t^f+\zeta_t,u_t^s)+ w_t = f(x_t^f,u_t^a)+w_t+\sigma_t,
\end{split}
\end{equation}
\begin{equation}\label{eq:error_b}
\begin{split}
 y_t^{c,a}=&h(x_t^a)+v_t-h(\hat{x}_t^a)+h(\hat{x}_t^a-s_t)\\
 =&h(x_t^f)+h(x_t^f+\zeta_t)-h(x_t^f)+h(x_t^a)-h(\hat{x}_t^a)+v_t\\
=& h(x_t^f) + v_t +\sigma'_t     
\end{split}
\end{equation}
\end{subequations}
with the initial condition $x_0^f=x_0^a-s_0$, where we also added and subtracted $f(x_t^f,u_t^a)$ and $h(x_t^f)$ from~\eqref{eq:error_a} and~\eqref{eq:error_b}, respectively. We also define $\sigma_t\delequal f(x_t^f+\zeta_t,u_t^s)-f(x_t^f,u_t^a)+f(x_t^a,u_t^a)-f(\hat{x}_t^a,u_t^s)$ and  $\sigma'_t \delequal h(x_t^f+\zeta_t)-h(x_t^f)+h(x_t^a)-h(\hat{x}_t^a)$. Since the controller dynamics~\eqref{eq:control} is IES and the exogenous input to~\eqref{eq:control_rec_d} and~\eqref{eq:attack_trajec_2} is the same (i.e., $y_t^{c,a}$), we have $\Vert \mathpzc{X}_{t}^a-\mathpzc{X}_{t}^s \Vert \leq \kappa_s\lambda_s^{t}\Vert \mathpzc{X}_{0}^a-\mathpzc{X}_{0}^s \Vert$ for $t\in \mathbb{Z}_{\geq 0}$. Since the function $h_c$ is Lipschitz, we get $\Vert u_t^a-u_t^s\Vert \leq L_{h_c} \kappa_s\lambda_s^{t}\Vert \mathpzc{X}_{0}^a-\mathpzc{X}_{0}^s \Vert$ for $t\in \mathbb{Z}_{\geq 0}$. 

Now, using the assumption that the functions $f$ and $h$ are Lipschitz, we can find a bound on the norm of $\sigma$ and $\sigma'$ as
\begin{subequations}
\begin{equation}\label{ineq:sigma}
\begin{split}
\Vert \sigma_t \Vert &\leq \Vert f(x_t^f+\zeta_t,u_t^s)-f(x_t^f,u_t^a)+f(x_t^a,u_t^a)-f(\hat{x}_t^a,u_t^s)\Vert\\ &\leq 2L_f(\Vert \zeta_t\Vert +L_{h_c} \kappa_s\lambda_s^{t}\Vert \mathpzc{X}_{0}^a-\mathpzc{X}_{0}^s \Vert)
\end{split}
\end{equation}
\begin{equation}
\Vert \sigma'_t \Vert \leq \Vert h(x_t^f+\zeta_t)-h(x_t^f)+h(x_t^a)-h(\hat{x}_t^a)\Vert
\leq 2L_h\Vert \zeta_t\Vert
\end{equation}
\end{subequations}

Now, if we define $\mathbf{X}^f_t=\begin{bmatrix}x_t^f\\\mathpzc{X}_{t}^a\end{bmatrix}$, it holds that
\begin{equation}\label{eq:closed_attack_2}
\mathbf{X}^f_{t+1}=F(\mathbf{X}^f_{t},\mathbf{W}'_t).
\end{equation}
with $\mathbf{X}^f_{0}=\begin{bmatrix}x_0^f\\\mathpzc{X}_{0}^a\end{bmatrix}$ and $\mathbf{W}'_t=\begin{bmatrix}w_{t}+\sigma_t\\v_t+\sigma'_t\end{bmatrix}$. Since $x_0^a=x_0$, it holds that $\mathbf{X}_{0}-\mathbf{X}^f_{0}=\begin{bmatrix}s_{0}\\\mathpzc{X}_{0}-\mathpzc{X}_{0}^a \end{bmatrix}$ and similar to Theorem~\ref{thm:PAt}, we have $\Vert \mathbf{X}_{0}-\mathbf{X}^f_{0} \Vert \leq (1+L_{f_c} L_h)\Vert s_0\Vert$. On the other hand, since the closed-loop system~\eqref{eq:closed_attack} is 
IISS, we~have 
\begin{equation}\label{eq:error_bound_2}
\begin{split}
\Vert \mathbf{X}&(t,\mathbf{X}_{0},\mathbf{W})-\mathbf{X}^f(t,\mathbf{X}^f_{0},\mathbf{W}')\Vert \leq \kappa \Vert \mathbf{X}_{0}-\mathbf{X}^f_{0}\Vert \lambda^{t}\\
&+\gamma(\Vert \mathbf{W}-\mathbf{W}'\Vert_{\infty}) \leq \kappa (1+L_{f_c} L_h)\Vert s_{0}\Vert \lambda^{t}\\
&+\gamma(\Vert \sigma_t\Vert_{\infty}+\Vert \sigma'_t\Vert_{\infty})\leq \kappa (1+L_{f_c} L_h)\Vert s_0\Vert \lambda^{t}\\
&+\gamma(2(L_f+L_h)\Vert \zeta_t\Vert+2L_fL_{h_c}\kappa_s\lambda_s^{t}\Vert \mathpzc{X}_{0}^a-\mathpzc{X}_{0}^s \Vert),
\end{split}
\end{equation}
where the term $\kappa (1+L_{f_c} L_h)\Vert s_{0}\Vert \lambda^{t}$ converges exponentially fast to zero. For simplicity of the notation, we denote the $\gamma$ function above simply by $\gamma(t)$ as it is a function of time. The term  $\gamma(t)$ has positive value and will approach to zero as both $\Vert \zeta_t\Vert$ and $\Vert \mathpzc{X}_{0}^a-\mathpzc{X}_{0}^s \Vert)$ approach to zero. Therefore, the trajectories of $\mathbf{X}$ (i.e., the system without attack)  and $\mathbf{X}^f$ are within a bounded distance from each other. We now use these results to show that the generated attack sequence satisfies the $\epsilon$-stealthiness condition. Due to the space limit, we only show here the stealthiness for state estimation method using Case 1. The proof for Case 2 can be found in Appendix. 

Before 
finishing the proof, we  introduce the lemma (proof in Appendix), and 
define $\mathbf{Z}_{t}=\begin{bmatrix}x_{t}\\y_t^{c}\end{bmatrix}$ and $\mathbf{Z}_{t}^f=\begin{bmatrix}x_t^f\\y_t^{c,a}\end{bmatrix}$. 

\begin{lemma}\label{lemma:determin}
Assume the sequence of $\mathbf{Z}_{-T}^{-1},\mathbf{Z}^f_{0}:\mathbf{Z}^f_{t}$
is given for any $t\geq 0$ and $s_0$, $\mathpzc{X}_{-T}$ and $\mathpzc{X}_0^s$ are chosen deterministically. Then, the signals $u_{t}^a$, $u_{t}^s$, $y_{t}^a$, $\hat{x}_{t+1}^a$ and $s_{t+1}$ are also uniquely (deterministically) obtained.
\end{lemma}
\vspace{+7pt}

Now, using the monotonicity property of the KL-divergence from Lemma~\ref{lemma:mon} it holds that
\begin{equation}\label{ineq:5_2}
\begin{split}
KL\big(&\mathbf{Q}(Y_{-T}^{-1},Y_{0}^a:Y_t^a)||\mathbf{P}(Y_{-T}:Y_t)\big)  \\
&\leq  KL\big(\mathbf{Q}(\mathbf{Z}_{-T}^{-1},\mathbf{Z}_{0}^f:\mathbf{Z}_{t}^f)||\mathbf{P}(\mathbf{Z}_{-T}:\mathbf{Z}_{t})\big),
\end{split}
\end{equation}
Then, we apply the chain-rule property of KL-divergence on the right-hand side of 
\eqref{ineq:5_2} to obtain the following
\begin{equation}\label{ineq:7}
\begin{split}
&KL\big(\mathbf{Q}(\mathbf{Z}_{-T}^{-1},\mathbf{Z}_{0}^f:\mathbf{Z}_{t}^f)||\mathbf{P}(\mathbf{Z}_{-T}^{-1},\mathbf{Z}_{0}:\mathbf{Z}_{t})\big)\\
&= KL\big(\mathbf{Q}(\mathbf{Z}_{-T}^{-1})||\mathbf{P}(\mathbf{Z}_{-T}^{-1})\big)+\\
&\quad \quad \quad \quad KL\big(\mathbf{Q}(\mathbf{Z}_{0}^f:\mathbf{Z}_{t}^f|\mathbf{Z}_{-T}^{-1})||\mathbf{P}(\mathbf{Z}_{0}:\mathbf{Z}_{t}|\mathbf{Z}_{-T}^{-1})\big)\\
&=KL\big(\mathbf{Q}(\mathbf{Z}_{0}^f:\mathbf{Z}_{t}^f|\mathbf{Z}_{-T}^{-1})||\mathbf{P}(\mathbf{Z}_{0}:\mathbf{Z}_{t}|\mathbf{Z}_{-T}^{-1})\big);
\end{split}
\end{equation}
where we used the fact that the 
KL-divergence of two identical distributions (i.e., $\mathbf{Q}(\mathbf{Z}_{-T}^{-1})$ and $\mathbf{P}(\mathbf{Z}_{-T}^{-1})$ is zero.

Applying the chain-rule property of 
KL-divergence to~\eqref{ineq:2} results in
\begin{equation}\label{ineq:8}
\begin{split}
& KL\big(\mathbf{Q}(\mathbf{Z}_{0}^f:\mathbf{Z}_{t}^f|\mathbf{Z}_{-T}^{-1})|| \mathbf{P}(\mathbf{Z}_{0}:\mathbf{Z}_{t}|\mathbf{Z}_{-T}^{-1})\big) \\
&= \sum_{k=0}^{t}\Big\{KL\big(\mathbf{Q}(x_k^f|\mathbf{Z}_{-T}^{-1},\mathbf{Z}_{0}^f:\mathbf{Z}_{k-1}^f)||\mathbf{P}(x_{k}|\mathbf{Z}_{-T}:\mathbf{Z}_{k-1})\big)\\
&+ KL\big(\mathbf{Q}(y_k^{c,a}|x_k^f,\mathbf{Z}_{-T}^{-1},\mathbf{Z}_{0}^f:\mathbf{Z}_{k-1}^f)||\mathbf{P}(y_{k}|x_k,\mathbf{Z}_{-T}:\mathbf{Z}_{k-1})\big)\Big\}.
\end{split}
\end{equation}

Given $\mathbf{Z}_{-T}:\mathbf{Z}_{k-1}$, the distribution of $x_k$ is a Gaussian with mean $f(x_{k-1},u_{k-1})$ and covariance $\mathbf{R}_w$. On the other hand, using~\eqref{eq:error_a} and~\eqref{eq:x_f_def} we have 
\begin{equation}\label{eq:fake_state_traj}
\begin{split}
x_k^f=&f(x_{k-1}^f+s_{k-1},u_{k-1}^a)-f(\hat{x}_{k-1}^a,u_{k-1}^s)\\
&+f(\hat{x}_{k-1}^a-s_{k-1},u_{k-1}^s)+w_{k-1}.
\end{split}
\end{equation}

Using Lemma~\ref{lemma:determin}, given  $\mathbf{Z}_{-T}^{-1},\mathbf{Z}_{0}^f:\mathbf{Z}_{k-1}^f$ and the fact that $s_0$, $\mathpzc{X}_{-T}$ and $\mathpzc{X}_0^s$ are deterministic, $\hat{x}_{k-1}^a$, $s_{k-1}$, $u_{k-1}^s$ and $u_{k-1}^a$ can be deterministically obtained. 
Therefore, the distribution of $x_k^f$ given $\mathbf{Z}_{-T}^{-1},\mathbf{Z}_{0}^f:\mathbf{Z}_{k-1}^f$ is a Gaussian with mean $f(x_{k-1}^f+s_{k-1},u_{k-1}^a)-f(\hat{x}_{k-1}^a,u_{k-1}^s)+f(\hat{x}_{k-1}^a-s_{k-1},u_{k-1}^s)$ and covariance $\mathbf{R}_w$. Using~\eqref{eq:free_trajec2} and~\eqref{eq:fake_state_traj}, it holds that 
\begin{equation}\label{eq:mean_diff}
\begin{split}
x_k^f-x_k=&f(x_{k-1}^f+s_{k-1},u_{k-1}^a)-f(\hat{x}_{k-1}^a,u_{k-1}^s)\\
&+f(\hat{x}_{k-1}^a-s_{k-1},u_{k-1}^s)-f(x_{k-1},u_{k-1}),
\end{split}
\end{equation}
for all $0\leq k \leq t$. Now, according to Lemmas~\ref{lemma:chain}, ~\ref{lemma:Guassian}, ~\ref{lemma:maximum} and equation~\eqref{eq:mean_diff}, for all $0\leq k \leq t$ we have
\begin{equation}\label{ineq:9}
\begin{split}
&KL\big(\mathbf{Q}(x_k^f|\mathbf{Z}_{-T}:\mathbf{Z}_{k-1}^f)||\mathbf{P}(x_{k}|\mathbf{Z}_{-T}:\mathbf{Z}_{k-1})\big) \\
&= \mathbb{E}_{\mathbf{Q}(x_k^f,\mathbf{Z}_{-T}^{-1},\mathbf{Z}_{0}^f:\mathbf{Z}_{k-1}^f)}\bigg\{(x_k-x_k^f)^T \mathbf{R}_w^{-1}  (x_k-x_k^f)\bigg\}\\
&\leq \lambda_{max}(\mathbf{R}_w^{-1})\Vert x_k-x_k^f \Vert^2.
\end{split}
\end{equation}
Given $x_k$, the distribution of $y_k$ is a Gaussian with mean $h(x_k)$ and covariance $\mathbf{R}_v$. On the other hand, from Lemma~\ref{lemma:determin}, we conclude that given $x_k^f,\mathbf{Z}_{-T}^{-1},\mathbf{Z}_{0}^f:\mathbf{Z}_{k-1}^f$ is Gaussian with mean $h(x_k^f)+\sigma'_k$ and covariance $\mathbf{R}_v$. Thus, using Lemmas~\ref{lemma:chain}, ~\ref{lemma:Guassian}, ~\ref{lemma:maximum} and~\eqref{eq:mean_diff}, for all $0\leq k \leq t$ we have
\begin{equation}\label{ineq:11}
\begin{split}
&KL\big(\mathbf{Q}(y_k^{c,a}|x_k^f,\mathbf{Z}_{-T}^{-1},\mathbf{Z}_{0}^f:\mathbf{Z}_{k-1}^f)||\mathbf{P}(y_{k}|x_k)\big)\\
&\quad =\mathbb{E}_{\mathbf{Q}(x_k^f,\mathbf{Z}_{-T}^{-1},\mathbf{Z}_{0}^f:\mathbf{Z}_{k-1}^f)}\Big\{\big(h(x_k^f)+\sigma'_k-h(x_k)\big)^T\\&\quad \quad \times  \mathbf{R}_v^{-1}  \big(h(x_k^f)+\sigma'_t-h(x_k)\big)\Big\} \leq  \lambda_{max}(\mathbf{R}_v^{-1}) \\
&\quad \quad \times \Big(L_h^2\Vert x_k-x_k^f \Vert^2+2L_h^2b_{\zeta}\Vert x_k-x_k^f \Vert+L_h^2b_{\zeta}^2\Big).
\end{split}
\end{equation}

Combining~\eqref{ineq:5_2}-\eqref{ineq:11} results in
\begin{equation*}
\begin{split}
KL&\big(\mathbf{Q}(Y_{-T}^{-1},Y_{0}^a:Y_t^a)||\mathbf{P}(Y_{-T}:Y_t)\big) \\
\leq &  \sum_{k=0}^{t} \Big\{\lambda_{max}(\mathbf{R}_w^{-1})\Vert x_k-x_k^f \Vert^2+\lambda_{max}(\mathbf{R}_v^{-1}) \\
& \, \times \Big(L_h^2\Vert x_k-x_k^f \Vert^2+2L_h^2b_{\zeta}\Vert x_k-x_k^f \Vert+L_h^2b_{\zeta}^2\Big)\Big\}. 
\end{split}
\end{equation*}
It is straightforward to verify that
\begin{subequations}
\begin{equation*}
\begin{split}
&\sum_{k=0}^{t} \Vert x_k-x_k^f \Vert^2 \leq \sum_{k=0}^{t} \big(\kappa (1+L_{f_c} L_h)\Vert s_0\Vert \lambda^{k} +\gamma(k)\big)^2\\
&=\sum_{k=0}^{t}\Big(\kappa^2 (1+L_{f_c} L_h)^2 \Vert s_0\Vert^2 \lambda^{2k} \\
&\qquad \qquad+2 \kappa (1+L_{f_c} L_h) \Vert s_0\Vert \lambda^{k} \gamma(k) + \gamma^2(k)\Big)\\
&\quad \leq \frac{\kappa^2 (1+L_{f_c} L_h)^2 \Vert s_0\Vert^2 }{1-\lambda^2}+2 \frac{\kappa (1+L_{f_c} L_h) \Vert s_0\Vert \gamma(0) }{1-\lambda} \\
&\qquad \qquad + \sum_{k=0}^{t} \gamma^2(k);
\end{split}
\end{equation*}

\begin{equation*}
\begin{split}
\sum_{k=0}^{t} \Vert x_k-x_k^f \Vert \leq & \sum_{k=0}^{t} \big(\kappa (1+L_{f_c} L_h)  \Vert s_0\Vert \lambda^{k} +\gamma(k)\big)\\
\leq &\frac{\kappa (1+L_{f_c} L_h) \Vert s_0\Vert }{1-\lambda}
+  \sum_{k=0}^{t} \gamma(k).
\end{split}
\end{equation*}
\end{subequations}
Therefore, we get

\begin{equation}\label{ineq:final_2}
\begin{split}
K&L\big(\mathbf{Q}(Y_{-T}^{-1},Y_{0}^a:Y_t^a)||\mathbf{P}(Y_{-T}:Y_t)\big) \leq\\
 &  \frac{\kappa^2 (1+L_{f_c} L_h)^2\Vert s_0\Vert^2 }{1-\lambda^2} \big(\lambda_{max}(\mathbf{R}_w^{-1}) + L_h^2\lambda_{max}(\mathbf{R}_v^{-1})\big)+ \\
& \frac{\kappa(1+L_{f_c} L_h) \Vert s_0\Vert }{1-\lambda} \big(2\gamma(0)\lambda_{max}(\mathbf{R}_w^{-1}) + L_h^2\lambda_{max}(\mathbf{R}_v^{-1})(2\gamma(0)+b_{\zeta})\big) \\
&+ \sum_{k=0}^{t} \gamma^2(k) \big(\lambda_{max}(\mathbf{R}_w^{-1}) + L_h^2\lambda_{max}(\mathbf{R}_v^{-1})\big)\\
&+\sum_{k=0}^{t} \gamma(k) \big(2L_h^2b_{\zeta}\lambda_{max}(\mathbf{R}_v^{-1})\big)+\sum_{k=0}^{t} b_{\zeta}^2L_h^2\lambda_{max}(\mathbf{R}_v^{-1}).
\end{split}
\end{equation}
Denoting the right-hand side of the above inequality by $b_{\epsilon}$ and using Theorem~\ref{thm:stealthy}, the system~$\Sigma_{\text{\rom{2}}}$ will be $\epsilon$-stealthy attackable with $\epsilon=\sqrt{1-e^{-b_{\epsilon}}}$.

Now, we show that the attack~\eqref{eq:attack_dynamics} is also $\alpha$-impactful. Using~\eqref{eq:attack_trajec_2} and~\eqref{eq:error_a} we have
\begin{equation}
\begin{split}
x_{t+1}^a&=f(x_t^a,u_t^a)+ w_t=f_u(x_t^a,U_t^a)\\
x_{t+1}^f&= f(x_t^f,u_t^a)+w_t+\sigma_t=f_u(x_t^f,{U'}_t^a)
\end{split}
\end{equation}
with $U_t^a=\begin{bmatrix}{u_t^a}^T&w_t^T\end{bmatrix}^T$and ${U'}_t^a=\begin{bmatrix}{u_t^a}^T&(w_t+\sigma_t)^T\end{bmatrix}^T$, for $t\in \mathbb{Z}_{\geq 0}$. Using the assumption that the system is IISU, there exists a nonzero $s_0$ such that $\Vert x^a(t,x_0^a,U^a)-x^f(t,x_0^a-s_0,{U'}^a)\Vert \geq M$ holds for any $M>0$ in $t\geq t'$ for some $t'$ given $\Vert {U}_t^a- {U'}_t^a \Vert \leq \Vert \sigma_t\Vert \leq 2L_f(b_{\zeta} +L_{h_c} \kappa_s\lambda_s^{t}\Vert \mathpzc{X}_{0}^a-\mathpzc{X}_{0}^s \Vert)$ for all $t\in \mathbb{Z}_{\geq 0}$. Following the same procedure as in Theorem~\ref{thm:PAt}, we can show that it holds that 
\begin{equation*}
\Vert x^a(t,x_0^a,U^a)-x(t,x_0-s_0,U)\Vert
\geq  M-\kappa (1+L_{f_c} L_h) \Vert s_0\Vert-\gamma(0), 
\end{equation*}
where by choosing $M\geq \alpha+\kappa (1+L_{f_c} L_h)\Vert s_0\Vert+\gamma(0)$, the attack will be $\alpha$-impactful. Therefore, the system is ($\epsilon$,$\alpha$)-attackable. 
\end{proof}
\vspace{4pt}
Unlike Theorem~\ref{thm:PAt} where the `fake' state converges exponentially fast to the attack-free state, under the conditions~of Theorem~\ref{thm:PA_2}, the `fake' state will be within a  bounded distance from the attack-free states. Intuitively, depending on the magnitude of deviation $\gamma$, different levels of stealthiness guarantees will be obtained, where for smaller values of $b_{\zeta}$ in Case 1 (or $b'_{\zeta}$ in Case 2), $\lambda_s$ and $\Vert \mathpzc{X}_0^s-\mathpzc{X}_0^a \Vert$, the deviation is smaller and thus, the attack will have stronger stealthier guarantees. The IES condition in Theorem~\ref{thm:PA_2} ensures the estimated input control $u_t^s$ converges to actual input control $u_t^a$ exponentially fast. 

The stealthiness can also be verified directly with the right-hand side of the inequalities~\eqref{ineq:final_2} and~\eqref{ineq:final_2_2} in Theorem~\ref{thm:PA_2}, respectively. These terms indicate $b_{\epsilon}$ for $\epsilon=\sqrt{1-e^{-b_{\epsilon}}}$ for each of Cases 1 and 2. As long as the attacker's estimate of states is more accurate (i.e., $b_{\zeta}$ or $b'_{\zeta}$ is smaller), then $\epsilon$ is smaller. However, there is no guarantee that $\epsilon$ approaches to zero when the estimation error is large. On the other hand, having the initial condition $\mathpzc{X}_0^s$ close to $\mathpzc{X}_0^a$ would help the attack be stealthier as the function $\gamma(k)$ would be smaller. Ideally, when $s_0$ can be chosen arbitrarily close to zero and both $b_{\zeta}$ (or $b'_{\zeta}$ ), $\Vert \mathpzc{X}_0^s-\mathpzc{X}_0^a \Vert$ approach to zero, then $\epsilon$ will be close to zero.    

We can also observe that having larger noise covariance $\mathbf{R}_w$ and $\mathbf{R}_v$ would help the attacker to have a stealthier attack as $\lambda_{max}(\mathbf{R}_w^{-1})$ and $\lambda_{max}(\mathbf{R}_v^{-1})$ would be smaller. However, if the attacker only relies on the system's sensor measurements for state estimation $\hat{x}_t^a$ (i.e., Case 1), having larger levels of noise would also cause larger $b_{\zeta}$ that has negative impact on attack stealthiness. To avoid this problem, the attacker might use the side information (Case 2) to have smaller estimation error independent of system's noise profile.

\subsection{Vulnerability of Nonlinear, Input Affine Systems} \label{sec:affine}

The model~\eqref{eq:plant} is very general and as a result it requires the attacker to have knowledge about the states and input control (system~$\Sigma_{\text{\rom{1}}}$) or their estimates (system~$\Sigma_{\text{\rom{2}}}$) to design the stealthy attacks. However, as we show in this section,  `simpler' dynamical models result in some relaxations on the required knowledge for the attacker to design stealthy attacks. 

For example, consider the plants that can be modeled as\footnote{With some abuse of the notation we again use $f$ as state transition function.}
\begin{equation}\label{eq:input_affine}
x_{t+1}=f(x_t)+Bu_t+w_t. 
\end{equation}
where $x$, $u$ and $w$ are defined as before and $B\in \mathbb{R}^{m\times n}$. Such formulation can include systems where the impact of the control input on states is weighted by a constant matrix or systems with format of $x_{t+1}=f(x_t)+g(x_t)u_t+w_t$ where the function $g(x)$ has very small Lipschitz constant. 

Let consider the system~$\Sigma_{\text{\rom{1}}}$ with  plant model~\eqref{eq:input_affine}. Therefore, when the conditions of Theorem~\ref{thm:PAt} holds such system will be ($\epsilon,\alpha$) attackable. However, generating the attack sequence using~\eqref{eq:attack_seq} for such system will be as
\begin{equation}
\begin{split}
s_{t+1}&=f(x_t^a)+Bu_t^a-f(x_t^a-s_t)-Bu_t^a \\ &=f(x_t^a)-f(x_t^a-s_t)\\
a_t&=h(x_t^a-s_t)-h(x_t^a),
\end{split}
\end{equation}
which means the attacker does not need to have access to input control during the attack ($u_t^a$) and such requirement in attack model~${\text{\rom{1}}}$ can be relaxed. Moreover, the attacker does not need to have knowledge about the matrix $B$. Similarly, for system~$\Sigma_{\text{\rom{2}}}$ with input affine plant model~\eqref{eq:input_affine}, the attack sequence generation~\eqref{eq:attack_dynamics} in Theorem~\ref{thm:PA_2} will also be relaxed~as 
\begin{equation}
\begin{split}
s_{t+1}&=f(\hat{x}_t^a)+Bu_t^s-f(\hat{x}_t^a-s_t)-Bu_t^s \\ &=f(\hat{x}_t^a)-f(\hat{x}_t^a-s_t)\\
a_t&=h(\hat{x}_t^a-s_t)-h(\hat{x}_t^a),
\end{split}
\end{equation}
Therefore, the attacker does not need to have access to the estimate of input ($u_t^s$), which as a result relaxes the assumption on having the knowledge about the controller dynamics in attack model~${\text{\rom{2}}}$ and the matrix $B$. Moreover, since the estimated input control is not needed, the IES assumption on the controller dynamics~\eqref{eq:control} can be removed from Theorem~\ref{thm:PA_2}. 

On the other hand, for such system the upper bound on $\Vert \sigma_t\Vert$ (\eqref{ineq:sigma} in Theorem~\ref{thm:PA_2}) will be smaller as we get $\Vert \sigma_t\Vert\leq 2L_f \Vert \zeta_t\Vert$. This will help the attacker to have smaller bound on $\epsilon$ and as a result have stealthier attack.


\subsection{Vulnerability Analysis of LTI Systems} 

We now derive sufficient conditions for 
($\epsilon,\alpha$)-successful attacks on LTI systems, with 
\eqref{eq:plant} and~\eqref{eq:control} 
 simplified as
\begin{equation}\label{eq:lti}
\begin{split}
{x}_{t+1} &= Ax_t+Bu_t+w_t,\quad y_t = Cx_t+v_t,\\
\mathpzc{X}_{t} &= A_c\mathpzc{X}_{t-1}+B_cy_t^c,\quad u_t  = C_c\mathpzc{X}_{t}.
\end{split}    
\end{equation}
LTI systems with any linear controller (e.g., LQG controllers) can be captured in the above form. 
The following lemma provides the conditions for IES and IU for the above 
LTI~system.

\begin{lemma}\label{lemma:LTI}
Consider 
an LTI dynamical system in the form of $x_{t+1}=Ax_t+Bd_t$. The system is IES if and only if all eigenvalues of the matrix $A$ are inside the unit circle. The system is 
IU if and only if $A$ has an unstable eigenvalue.
\end{lemma}
\begin{proof}
Let consider two trajectories $x(t,\xi_1,d)$ and $x(t,\xi_2,d)$ with initial conditions $\zeta_1$ and $\zeta_2$, respectively, and the input sequence of $d=d_0:d_{t-1}$. Now, let define $\Delta x_{t}=x(t,\xi_1,d)-x(t,\xi_2,d)$, which has the dynamics $\Delta x_{t+1}= A\Delta x_{t}$ with initial condition $\Delta x_{0}=\xi_1-\xi_2$. Hence, $\Delta x_t$  converge exponentially to zero if and only if the matrix $A$ has all eigenvalues inside the unit circle.  Now, let assume the matrix $A$ is unstable and for any trajectory  $x(t,\xi_1,d)$ consider another trajectory $x(t,\xi_2,d)$ with   $\xi_2=\xi_1+cq_i$ where $c$ is a constant and $q_i$ is the unstable eigenvector associate with unstable eigenvalue $\lambda_i$ of the matrix $A$. Now, we get $\Delta x_{t}=A^t\Delta x_{0}=c\lambda_i^tq_i$. Since $\vert \lambda_i\vert>1$, for any $M>0$, there exists a time step $t'$ such that $\Vert \Delta x_{t'}\Vert >M$. { It is also straightforward to show the converse direction.}   
\end{proof}

This allows us to directly capture conditions for stealthy yet effective attacks on LTI systems. 
\begin{corollary}
\label{cor:cor1}
The LTI system~\eqref{eq:lti} is ($\epsilon,\alpha$)-attackable for arbitrarily large $\alpha$ if the matrix $A$ is unstable and the closed-loop control system is asymptotically stable. 
\end{corollary}

\begin{proof}
The proof is directly obtained\textbf{} by combining Theorem~\ref{thm:PAt} and Lemma~\ref{lemma:LTI}. 
\end{proof}

Asymptotic stability of the closed-loop system is not a restrictive assumption as stability is commonly the weakest required performance guarantee for a control system. Matrix $A$ being unstable is a sufficient condition for satisfying ($\epsilon,\alpha$)-attackability when any set of sensors can be compromised. Note that the  ($\epsilon,\alpha$)-attackability condition for LTI systems with an optimal detector 
complies with the results from~\cite{mo2010false,jovanov_tac19} where LQG controllers with residue based detectors (e.g., $\chi^2$ detectors) are considered.

In this case, the false-data injection attack sequence design method from~\eqref{eq:attack_seq} reduces into a simple dynamical model
\begin{equation}\label{eq:attack_lin}
\begin{split}
s_{t+1}&=Ax_t^a+Bu_t^a - (A(x_t^a-s_t)+Bu_t^a)=As_t \\
a_t&=C(x_t^a-s_t)-C(x_t^a)=-Cs_t,
\end{split}
\end{equation}
which only requires knowledge of the matrices $A$ and~$C$ {unlike the works~\cite{mo2010false,jovanov_tac19,kwon2014analysis,zhang2020false} that assume the attacker has access to the full plant model information as well as the controller and Kalman filter gain.} Similarly, the attack sequence generated using the attack model~${\text{\rom{2}}}$ in~\eqref{eq:attack_dynamics} for LTI systems would have the same form as above.
As a result, the attacker does not need to have access to states, input control (or their estimates) and therefore, both attack models fall into the same attack model where the attacker only needs to know the matrices $A$ and $C$. 

As discussed, $s_t$ is a measure of the distance between $x_t^a$ and $x_t$; thus,  for impactful attacks, $s_t$ in~\eqref{eq:attack_lin} needs to dynamically grow over time, and $s_0$ needs to be in the subspace of unstable eigenvectors of matrix $A$. For example, if $q_i$ is the eigenvector associated with unstable eigenvalue $\lambda_i$, for $s_0=cq_i$ we get $s_t=c\lambda_i^tq_i$, where $c$ can be chosen by the attacker. Using Theorem~\ref{thm:PAt} we have $\Vert x_t^a-x_t\Vert \geq c\lambda_i^tq_i -  (1+\Vert A_c\Vert  \Vert C\Vert )c\Vert q_i\Vert $ resulting in large deviation in states for arbitrarily large time steps. 

On the other hand, if the condition of Corollary~\ref{cor:cor1} holds and the attack is generated using~\eqref{eq:attack_lin}, then the attack is $\epsilon$-stealthy. From Theorem~\ref{thm:PAt}, we have $\epsilon=\sqrt{1-e^{-b_{\epsilon}}}$ with 
\begin{equation*}
b_{\epsilon}=\frac{c^2\Vert q_i\Vert^2}{1-\lambda_{c}^2} \big(\lambda_{max}(\mathbf{R}_w^{-1})+\Vert C\Vert^2 \lambda_{max}(\mathbf{R}_v^{-1})\big)\delequal b_{\epsilon},
\end{equation*}
{where we used $s_0=cq_i$, and $\lambda_c$ is the decay rate of the closed-loop system} and all parameters are defined as before. Since all parameters above are characterized by the system except $c$, the attacker can choose the scalar $c$ arbitrarily close to zero to have arbitrarily nonzero small $\epsilon$.

\begin{remark}
We initially assumed that $\mathcal{K}=\mathcal{S}$; i.e., the attacker can compromise all sensors. However, when the system is LTI, the minimum subset of compromised sensors can be obtained as
$\min_{q_i\in \{q_1,...,q_f\}} \Vert \text{supp} (Cq_i)\Vert_0$, 
%
where $ \{q_1,...,q_f\}$ denotes the set of unstable eigenvectors of the matrix $A$, and $\text{supp}$ denotes the set of nonzero elements of the vector. 
\end{remark}  

%% file: Estimation.tex
\section{The Impact of Stealthy Attack} 
\label{sec:estimation}

In this section, we evaluate the impact of the stealthy attack from~\eqref{eq:attack_seq} on state estimation. Here, as an example we consider EKF; however, the results can be extended to any other nonlinear Luenberger type observer.  
When the sensors are attack free, the state estimates of~\eqref{eq:plant} are updated as
\begin{equation}
\begin{split}
\hat{x}_{t|t-1} &= f(\hat{x}_{t-1},{u}_{t-1}),\\
\hat{x}_t &= \hat{x}_{t|t-1} +L(y_t^c-h(\hat{x}_{t|t-1})),
\end{split}
\end{equation}
where $L$ is the observer gain that is assumed in steady state. 
We assume that the control input $u_t$ is obtained by the state estimate $\hat{x}_t$ as $u_t=K(\hat{x}_t)$. Thus, the  estimation dynamics~is
\begin{equation*}
\hat{x}_t =f(\hat{x}_{t-1},K(\hat{x}_{t-1})) +L(y_t^c-h(f(\hat{x}_{t-1},K(\hat{x}_{t-1})))), 
\end{equation*}
which is a special form of~\eqref{eq:control}. Now, the dynamics of the closed-loop system can be captured~as
\begin{equation}\label{eq:closed-loop_2}
\mathbf{X}_{t+1}'=F'(\mathbf{X}_t',\mathbf{W}_t)
\end{equation}
where we define the full state of the closed-loop  system as $\mathbf{X}'\delequal \begin{bmatrix}
{x}_t^T & \hat{x}_t^T \end{bmatrix}^T$, and exogenous disturbances as $\mathbf{W}_t\delequal \begin{bmatrix}
w_t^T&{v}_{t+1}^T&v_t^T \end{bmatrix}^T$ similar defined as in~\eqref{eq:closed-loop}. In what follows, we provide the condition such that there exists a sequence of $(\epsilon,\alpha)$-attacks for which the estimate of states during the attack converges exponentially fast to attack free case. 

\begin{theorem}\label{thm:estim}
If the closed-loop system~\eqref{eq:closed-loop_2} with attack model~${\text{\rom{1}}}$ is 
IES and the open loop system~\eqref{eq:input-state} is 
IU,
then there exists a sequence of $(\epsilon,\alpha)$-attacks such that $$\Vert \hat{x}_t^f-\hat{x}_t \Vert \leq \eta \lambda^t $$ for some positive $\lambda<1$ and some $\eta>0$. 
\end{theorem}

\begin{proof}
In the proof of Theorem~\ref{thm:PAt}, we showed that if the closed-loop system~\eqref{eq:closed-loop_2} is 
IES and the open-loop system~\eqref{eq:input-state} is 
IU, then the system is ($\epsilon,\alpha$-attackable), i.e., 
a sequence of attacks generated by~\eqref{eq:attack_seq}
will cause $\alpha$ deviation in the states while remaining $\epsilon$-stealthy. Now, we will show the impact of such attack on the state estimation. 

Similar to Theorem~\ref{thm:PAt}, by defining $x_t^f\delequal x_t^a-s_t$ we get
\begin{equation}\label{eq:fake_state_es}
\begin{split}
x_{t+1}^f=&f(x_t^f,K(\hat{x}_t^f))+ w_t^a,\quad y_t^{c,a}=h(x_t^f)+v_t^a,\\
\hat{x}_t^f =&f(\hat{x}_{t-1}^f,K(\hat{x}_{t-1}^f)) +L(y_t^{c,a}-h(f(\hat{x}_{t-1},K(\hat{x}_{t-1}))))
\end{split}
\end{equation}

On the other hand, if the system were not under attack during $t\in \mathbb{Z}_{\geq 0}$, the plant and the state estimation evolution would satisfy
\begin{equation}\label{eq:free_trajec_es}
\begin{split}
x_{t+1}=&f(x_t,K(\hat{x}_t))+ w_t,\quad y_t^c=h(x_t)+v_t,\\
\hat{x}_t =&f(\hat{x}_{t-1},K(\hat{x}_{t-1})) +L(y_t^{c}-h(f(\hat{x}_{t-1},K(\hat{x}_{t-1}))))
\end{split}
\end{equation}
Since the system and measurement noises are independent of the state, we can assume that $w_t^a=w_t$ and $v_t^{a}=v_t$. By defining $\mathbf{X}'^f_t \delequal \begin{bmatrix}
({x}_t^f)^T & (\hat{x}_t^f)^T \end{bmatrix}^T$, the dynamics~\eqref{eq:fake_state_es} and~\eqref{eq:free_trajec_es} can be written as $\mathbf{X}_{t+1}'^f=F'(\mathbf{X}_t'^f,\mathbf{W}_t)$ and $\mathbf{X}_{t+1}'=F'(\mathbf{X}_t',\mathbf{W}_t)$. Since both these dynamics are subject to the same input and the system dynamics is 
IES, there exist $\kappa>1$ and
$\lambda<1$ such that
\begin{equation}
\begin{split}
\Vert \mathbf{X}_t'^f - \mathbf{X}_{t}' \Vert 
\leq \kappa \Vert \mathbf{X}_0'^f - \mathbf{X}_{0}' \Vert \lambda^t
\end{split}
\end{equation}
with $ \mathbf{X}_0'^f - \mathbf{X}_{0}' = \begin{bmatrix}
x_0^f-x_0\\\hat{x}_0^f-\hat{x}_0
\end{bmatrix} =\begin{bmatrix}
s_0\\\hat{x}_0^f-\hat{x}_0
\end{bmatrix}$; 
thus, we have
\begin{equation}
\begin{split}
\hat{x}_0^f-\hat{x}_0 = L(y^{c,a}_0-y_0^c)=L(h(x_0-s_0)-h(x_0))
\end{split}
\end{equation}
where we used the fact that attack starts at $t=0$; i.e., $\hat{x}^f_{-1}=\hat{x}_{-1}$ and $x_0^a=x_0$. Therefore, we get $\Vert \mathbf{X}_0'^f - \mathbf{X}_{0}'\Vert \leq (1+\Vert L\Vert L_h)\Vert s_0\Vert$. As a result we get 
\begin{equation}
\begin{split}
\Vert \hat{x}_t^f - \hat{x}_{t} \Vert\leq \Vert \mathbf{X}_t'^f - \mathbf{X}_{t}' \Vert 
\leq \kappa (1+\Vert L\Vert L_h)\Vert s_0 \Vert \lambda^t,
\end{split}
\end{equation}
and finally, for $\eta =\kappa (1+\Vert L\Vert L_h)\Vert s_0 \Vert$, it holds that $\Vert \hat{x}_t^f-\hat{x}_t \Vert \leq \eta \lambda^t $. 
\end{proof}

For LTI plants where a linear Kalman filter is used for state estimation and the LQR controller is used to control the plant, the closed-loop system satisfies
\begin{align}\label{eq:LQG}
{x}_{t+1} &= Ax_t+Bu_t+w_t,\quad y_t = Cx_t+v_t,\nonumber\\
\hat{x}_{t} &= A\hat{x}_{t-1}+Bu_{t-1}+L(y_t^c-C(A\hat{x}_{t-1}+Bu_{t-1})),\nonumber\\ u_t  &= K\hat{x}_{t};
\end{align}
here, $K$ and $L$ are obtained by solving algebraic Riccati equations. 
When the system is observable and controllable, the obtained $K$ and $L$ are stabilizing the closed-loop system~\cite{burl1998linear}.

\begin{corollary}
Consider the closed-loop system~\eqref{eq:LQG} with pairs $(A,B)$ and $(A,C)$ controllable and observable, respectively. If matrix $A$ is unstable,  then there exists a sequence of $(\epsilon,\alpha)$-attacks such that $\Vert \hat{x}_t^f-\hat{x}_t \Vert \leq \eta \lambda^t$ for some positive $\lambda<1$ and some $\eta>0$. 
\end{corollary}

\begin{proof}
Since the system is controllable and observable, $K$ and $L$ in the LQG controller are stabilizing the closed-loop system~\eqref{eq:LQG}; thus, from Lemma~\ref{lemma:LTI} the closed-loop  system is 
IES. The rest of the proof follows the proof of Theorem~\ref{thm:estim}.
\end{proof}

Hence, for LTI systems with LQG controllers, the attack generated by~\eqref{eq:attack_lin} causes the state estimation during the attack to converge exponentially fast to attack-free state estimates. This result is important as it shows the changes in estimation of states are small due to the attack and the defender cannot detect the presence of attack by observing the state estimates.

%% file: Simulation.tex
\section{Simulation Results}
\label{sec:simulation}

We illustrate and evaluate our methodology for vulnerability analysis of nonlinear control systems 
on two case studies,
cart-pole and unmanned aerial  vehicles~(UAVs). 

\subsection{Cart-pole}
We use the dynamics of the cart-pole system 
\cite{florian2007correct}
\begin{equation} \label{eq:state_card}
\begin{split}
\ddot{\theta}&=\frac{g\sin\theta+\cos\theta\big(\frac{-F-M_pl\dot{\theta}^2\sin\theta}{M_c+M_p}\big)}{l\big(\frac{4}{3}-\frac{M_p\cos^2\theta}{M_c+M_p}\big)}\\
\ddot{x}&=\frac{F+M_pl(\dot{\theta}^2\sin\theta-\ddot{\theta}\cos\theta)}{M_c+M_p};
\end{split}
\end{equation}
here, $\theta$ is the pendulum angle from vertical, $x$ is the cart
position along the track, and $F$ is the control torque applied to the pivot, $g=9.81$ is the acceleration due to gravity, $M_c=1~\textit{Kg}$ is the mass of the cart and $M_p=0.1~\textit{Kg}$ is the mass of the pole. The length of the pole is $2l=1~\textit{m}$. We assume that only $x$ and $\theta$ are directly measured and the system is equipped with an 
EKF to estimate system states. 

The system and measurement noise are $\mathbf{R}_w=\sigma^2I$ and $\mathbf{R}_v=\sigma^2I$ where $I$ is identity matrix with suitable dimension. Moreover, 
a feedback full state controller that uses estimated states by the EKF was used to keep the pendulum inverted
around $\theta=0$ and $x=0$ equilibrium point. To detect the presence of attack, we consider two standard widely-used  IDs -- i.e., $\chi^2$ and  cumulative sum (CUSUM) detectors~\cite{umsonst2017security}. For each ID, the residue is obtained by comparing the current received sensor measurement and the expected measurement. We set the thresholds for both these detectors to $p^{FA}=0.002$. 

\begin{figure}
\centering
\subfloat[]{\label{fig:lti_attack_1}
\includegraphics[width=0.24\textwidth, height=3.4cm]{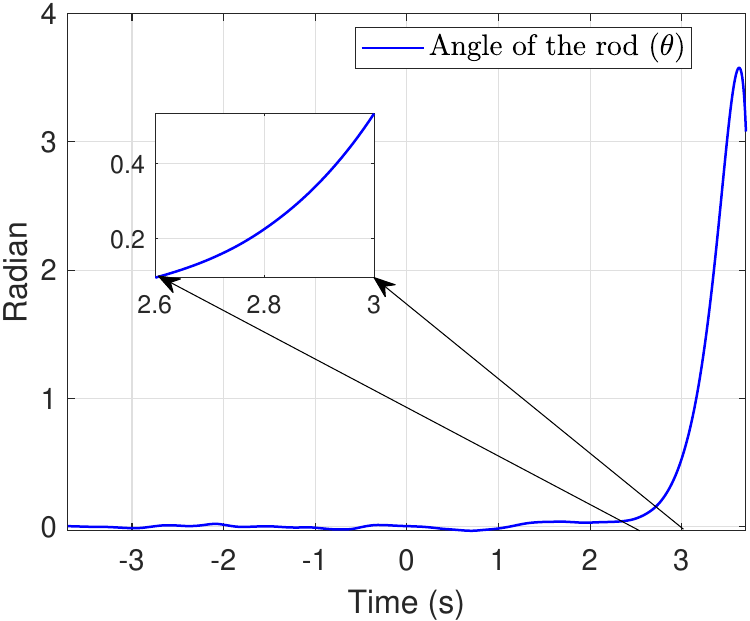}}
\subfloat[]{\label{fig:lti_attack_2} \includegraphics[width=0.24\textwidth, height=3.4cm]{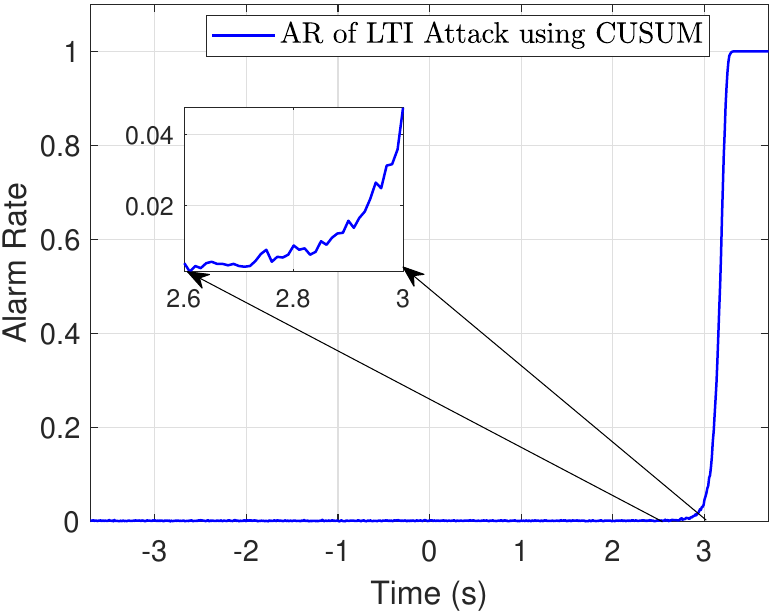}}
\vspace{-6pt}
\caption{(a) The angle of the pendulum rod over time while the LTI-based attack starts at time zero. (b) The alarm rate for LTI-based attack on Cart-pole system over 5000 experiments.}
\label{fig:lti_attack}
\end{figure}

Fig.~\ref{fig:lti_attack} 
shows ineffectiveness of LTI attacks on Cart-pole system with actual nonlinear dynamics~\eqref{eq:state_card} and $\sigma^2=0.001$. The LTI attack was designed according to~\eqref{eq:attack_lin}, where $A$ was obtained by linearizing the nonlinear dynamics~\eqref{eq:state_card} around the equilibrium point at zero. 
The attack starts at $t=0$ with initial condition $\Vert s_0\Vert =10^{-7}$. Fig.~\ref{fig:lti_attack_1} shows the evolution of rod's angle over time, while Fig.~\ref{fig:lti_attack_2} shows the average alarm rate for the CUSUM detector under the LTI-based attack. The zoomed area highlights that when the rod's angle starts deviating from zero, the the alarm rate of the ID increases over time. Specifically, for $0.4\, rad = 22.9^{\circ}$, $p^{TD}$ is almost ten times larger than $p^{FA}$. This demonstrates the limitation of the LTI-based attacks as the LTI approximation of a nonlinear model is only valid for the region around the equilibrium point and as the system moves away from equilibrium point, the approximation error significantly increases.

%

For the attack model~${\text{\rom{1}}}$, we used the attack sequence~\eqref{eq:attack_seq}. Fig.~\ref{fig:CP_fig1} shows $\epsilon$ versus different norm values of the initial condition $s_0$ for fixed $\sigma^2=0.001$. By increasing the norm of the initial condition, 
$\epsilon$ also increases making the attack likelier to be detected. However, choosing the initial condition close to zero (i.e., with a very small norm) has a side effect that it takes more time until the attack becomes effective as illustrated in Fig.~\ref{fig:CP_fig2}. Hence, there is a trade-off between attack stealthiness and time before it becomes effective. 
The impact of noise variance on attack stealthiness is also considered, in Fig.~\ref{fig:CP_fig3} where for a fixed $\Vert s_0\Vert = 10^{-8}$, we show that increasing the noise power $\sigma^2$ decreases the value of $\epsilon$. 

In Fig.~\ref{fig:CP_fig4} and~\ref{fig:CP_fig5}, we considered the impact of attacker's state estimation error, $b_{\zeta}$ in attack model~${\text{\rom{2}}}$, on the detection rate of CUSUM and $\chi^2$ detectors when $\Vert s_0\Vert = 10^{-8}$.  
Having inaccurate estimate with $b_{\zeta}=0.5$ can result in attack detection, while more accurate estimate ($b_{\zeta}=0.05$) improves attack stealthiness. 
Fig.~\ref{fig:CP_fig6} shows the evolution of rod's angle over time for different values of $b_{\zeta}$ --   
after initiating the attack at $t=0$, 
the rod angle increases over time in both cases.



\begin{figure*}[!t]
\begin{subfigure}{0.312\textwidth}
\includegraphics[width=.9\linewidth, height=3.4cm]{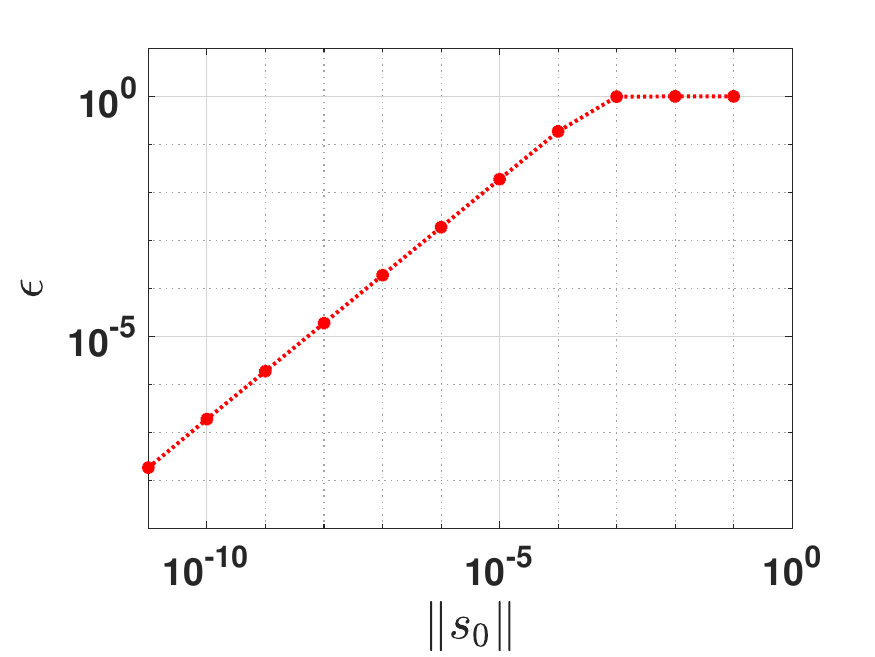}
\hspace{-6pt}
\caption{ }\label{fig:CP_fig1}
\end{subfigure}
\begin{subfigure}{0.312\textwidth}
\includegraphics[width=.9\linewidth, height=3.4cm]{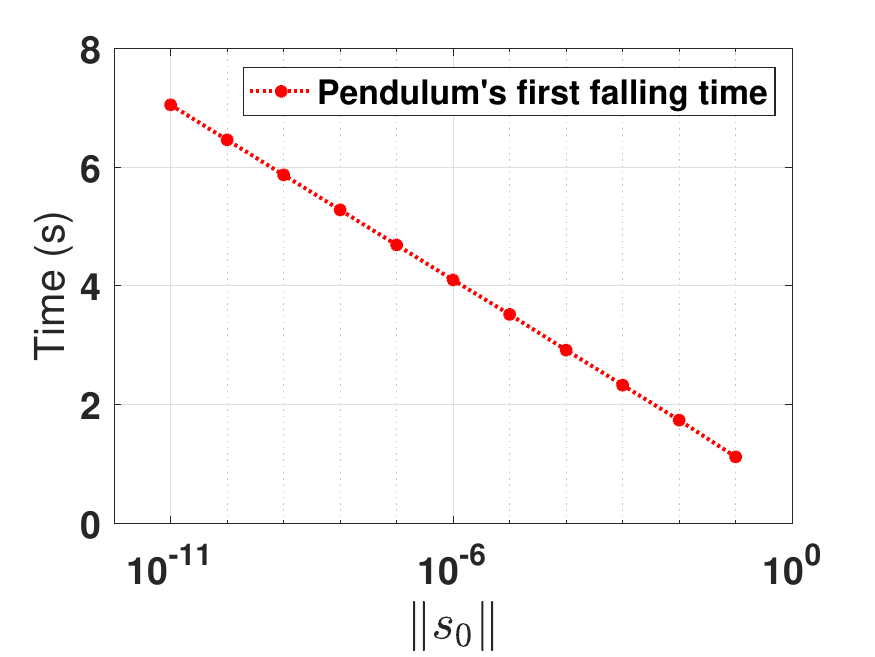}
\hspace{-6pt}
\caption{}\label{fig:CP_fig2}
\end{subfigure}
\begin{subfigure}{0.312\textwidth}
\includegraphics[width=.9\linewidth, height=3.4cm]{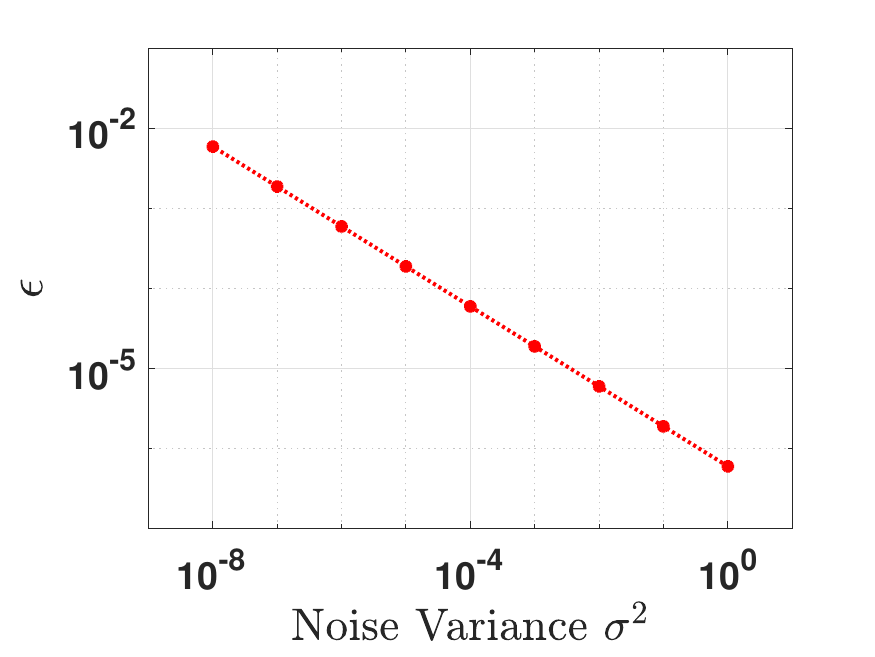}
\hspace{-6pt}
\caption{}\label{fig:CP_fig3}
\end{subfigure}
\begin{subfigure}{0.312\textwidth}
\includegraphics[width=.9\linewidth, height=3.4cm]{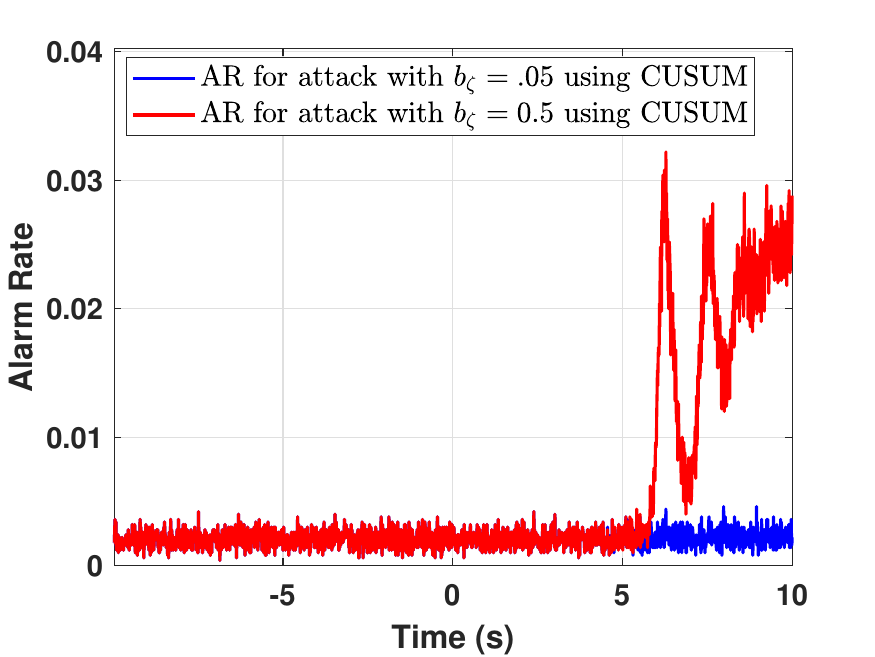}
\vspace{-6pt}
\caption{}\label{fig:CP_fig4}
\end{subfigure}
\begin{subfigure}{0.312\textwidth}
\includegraphics[width=.9\linewidth, height=3.4cm]{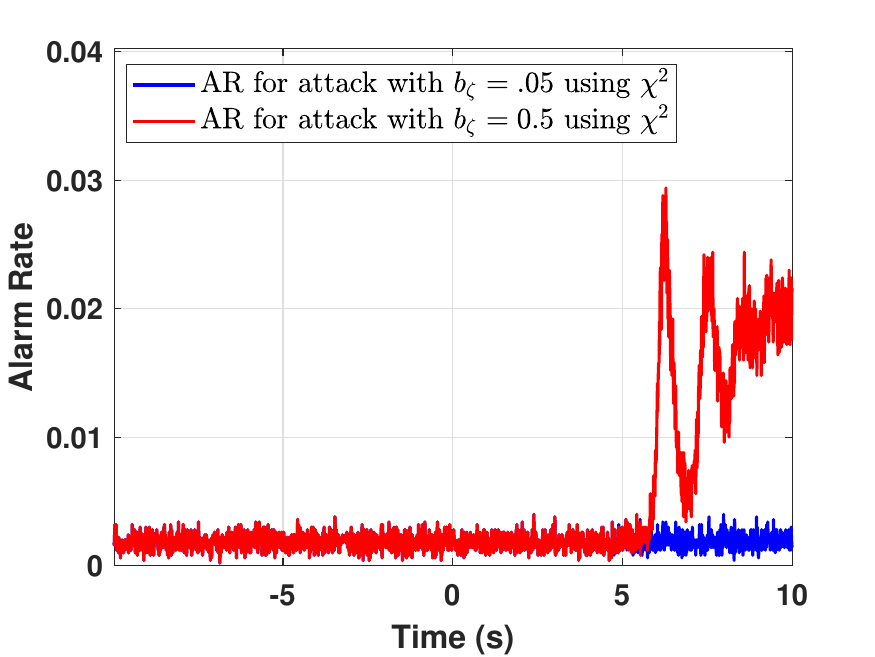}
\vspace{-6pt}
\caption{}\label{fig:CP_fig5}
\end{subfigure}
\begin{subfigure}{0.312\textwidth}
\includegraphics[width=.9\linewidth, height=3.4cm]{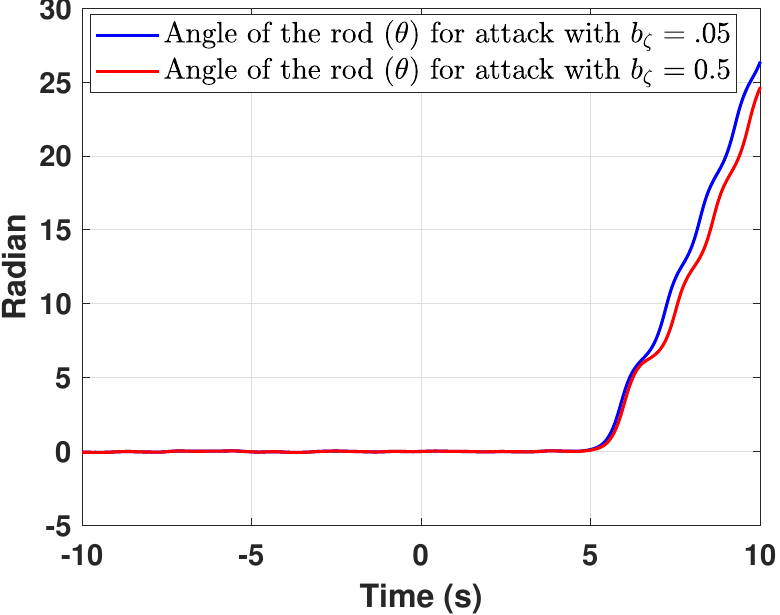}
\vspace{-6pt}
\caption{}\label{fig:CP_fig6}
\end{subfigure}
\vspace{-6pt}
\caption{(a,b)~ $\epsilon$ and falling time  of the pendulum rod versus $\Vert s_0\Vert$, for $\sigma^2=10^{-3}$. (c)~The impact of system and measurement noise power $\sigma^2$ on attack stealthiness $\epsilon$ for $\Vert s_0\Vert=10^{-8}$; (d,e)~The alarm rate for CUSUM and $\chi^2$ intrusion detectors before and after the attack (attack starts at time $t=0$) for different values of $b_{\zeta}$ over 5000 experiments for $\Vert s_0\Vert=10^{-8}$ and $\sigma^2=10^{-3}$; (f)~The trajectories of the pendulum rod angle versus time for different values of $b_{\zeta}$.}
\label{fig:CP_fig}
\vspace{-8pt}
\end{figure*}

\subsection{Unmanned Aerial Vehicles}

We also considered a quadrotor with complex highly nonlinear model from~\cite{bouabdallah2007full} that has 12 states $\begin{bmatrix} x ,\, y ,\, z ,\, \dot{x} , \, \dot{y},\, \dot{z} ,\, \phi ,\, \theta ,\,  \psi ,\,   \dot{\phi} ,\, \dot{\theta}, \, \dot{\psi}\end{bmatrix}^T$;  $x$, $y$ and $z$ represent the quadrotor position along the $X$, $Y$ and $Z$ axis, respectively, while $\dot{x}$, $\dot{y}$ and $\dot{z}$ are their velocity. 
 $\phi$, $\theta$ and $\psi$ are pitch, roll and yaw angles respectively, and $\dot{\phi}$, $\dot{\theta}$ and $\dot{\psi}$ represent their corresponding angular velocity. The system was 
discretized using Euler method with $T_s=0.01s$. 
The states  $\begin{bmatrix} x ,\, y ,\, z ,\, \phi ,\, \theta ,\, \psi,\dot{\phi} ,\, \dot{\theta}, \, \dot{\psi} \end{bmatrix}^T$ were measured 
with zero-mean Gaussian noise with the covariance matrix~$\mathbf{R}_v=0.001 I$. 

We assumed standard disturbance on the input modeled by system noise with zero mean Gaussian with the covariance matrix $\mathbf{R}_w=0.001 I$. 
The system employs an EKF to estimate the states and a PID controller 
to control the drone. 
To detect the presence of any abnormal behaviour,  CUSUM and $\chi^2$ IDs were deployed with the threshold fixed to $p^{FA}=0.002$. 

We considered the position control task~\cite{bouabdallah2007full}, where 
the drone takes off from (0,0,0) in global coordinate and reaches to predefined point in the space $(0,0,10 m)$ and stays there (see the black line trajectory in Fig.~\ref{fig:drone_trajectory}). Once the drone reaches to that point, the attack starts (it should be noted that there is no limitation for the attacker on the starting time of the attack and our assumption here is just to illustrate the results better). 

For attack model~${\text{\rom{1}}}$ with generated attack sequence as~\eqref{eq:attack_seq}, Fig.~\ref{fig:drone_eps} shows the impact of $\Vert s_0\Vert$ on attack stealthiness. The results again 
show that a larger norm of the initial condition $s_0$ results in an increased attack detection rate. 
For the attack model~${\text{\rom{1}}}$ where the attacker generates attack sequence by~\eqref{eq:attack_seq_2} using case 1 and $\Vert s_0\Vert =10^{-5}$ from $t=0$, 
Fig.~\ref{fig:drone_eps} shows the alarm rate of CUSUM and $\chi^2$ IDs for 50 seconds. 

We also showed that by fixing $\Vert s_0\Vert =10^{-5}$,  and changing the nonzero elements of the $s_0$ the attacker can control the direction of the drone's deviation (see Fig.~\ref{fig:drone_trajectory}). 
Assuming that the attack starts at the red square, Fig.~\ref{fig:drone_trajectory1} shows that by placing a positive nonzero value on the element $s_0$ associated with pitch angle ($\phi$) causes the drone to deviate $20~m$ in negative $Y$-axis in $50~s$. Such negative $s_0$ results on the drone's deviation in the positive direction of $Y$ axis. It also shows that putting the nonzero element in roll angle can cause the deviation along with $X$ axis. However, if the attacker uses the combination of both, the drone deviates in both $X$ and $Y$ axis ( Fig.~\ref{fig:drone_trajectory2}).

\begin{figure*}[!t]
\begin{subfigure}{0.312\textwidth}
\includegraphics[width=.9\linewidth, height=3.4cm]{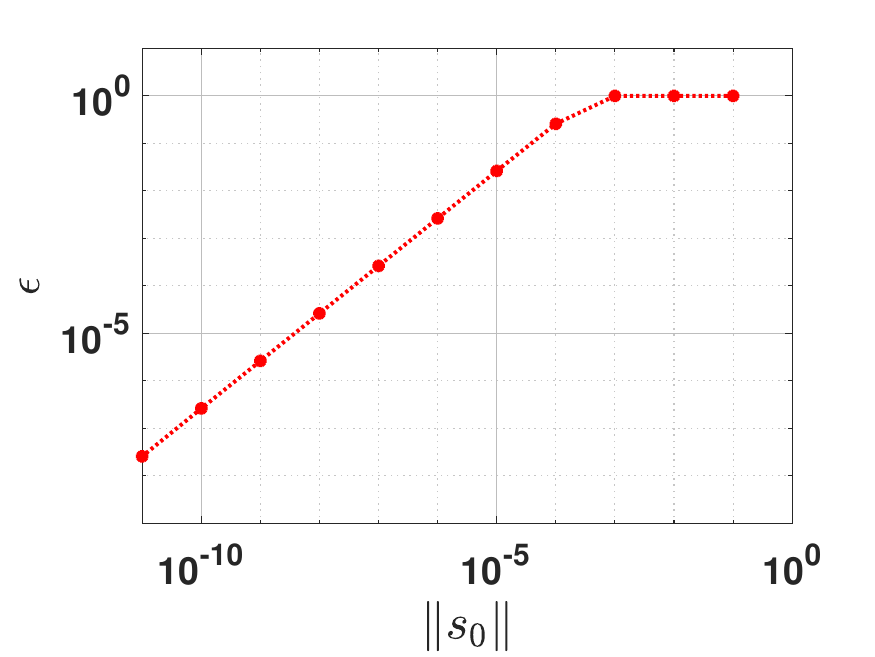}
\hspace{-6pt}
\caption{ }\label{fig:drone_eps}
\end{subfigure}
\begin{subfigure}{0.312\textwidth}
\includegraphics[width=.9\linewidth, height=3.4cm]{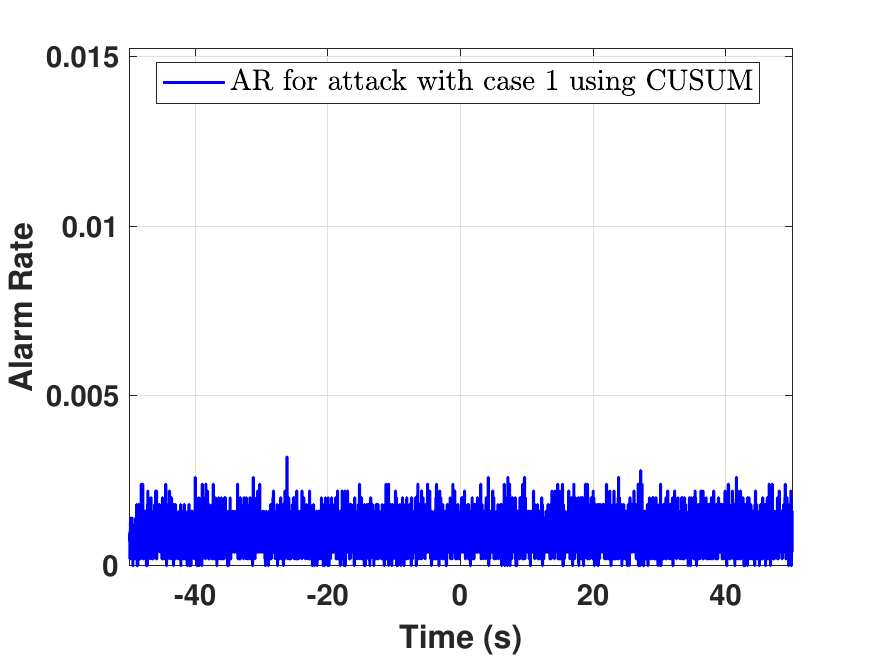}
\hspace{-6pt}
\caption{}\label{fig:drone_CUSUM}
\end{subfigure}
\begin{subfigure}{0.312\textwidth}
\includegraphics[width=.9\linewidth, height=3.4cm]{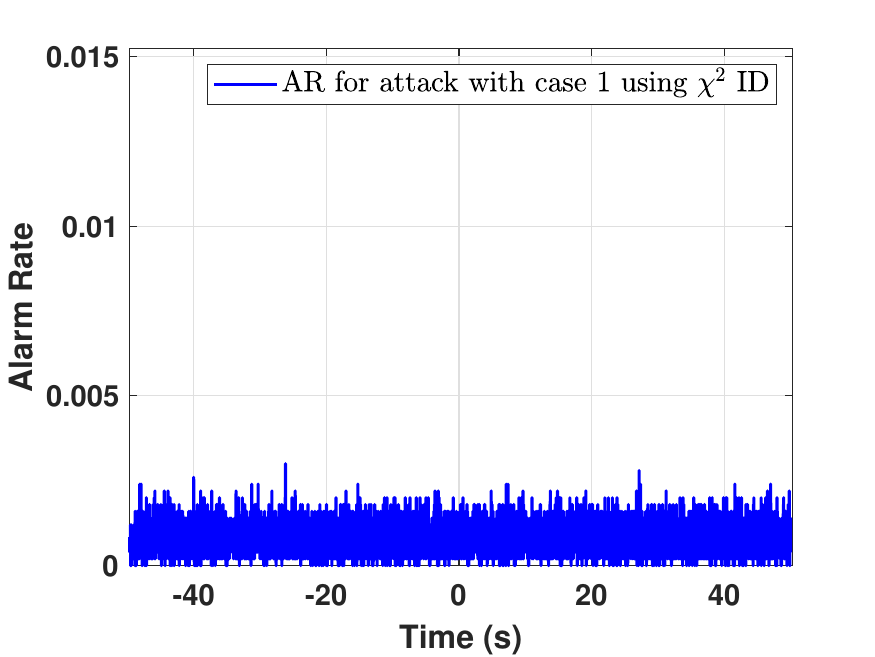}
\hspace{-6pt}
\caption{}\label{fig:drone_Chi}
\end{subfigure}
\vspace{-6pt}
\caption{(a)~ $\epsilon$ versus $\Vert s_0\Vert$, for $\sigma^2=10^{-3}$ in UAV case study. (b,c)~The alarm rate for CUSUM and $\chi^2$ intrusion detectors before and after the attack start time ($t=0$) for case 1 (the attacker uses system's sensors to estimate the states) over 5000 experiments for $\Vert s_0\Vert=10^{-5}$ and $\sigma^2=10^{-3}$.}
\label{fig:drone_stealthiness_fig}
\vspace{-8pt}
\end{figure*}

\begin{figure}
\centering
\subfloat[]{ \label{fig:drone_trajectory1}
\includegraphics[width=0.223\textwidth, height=3.4cm]{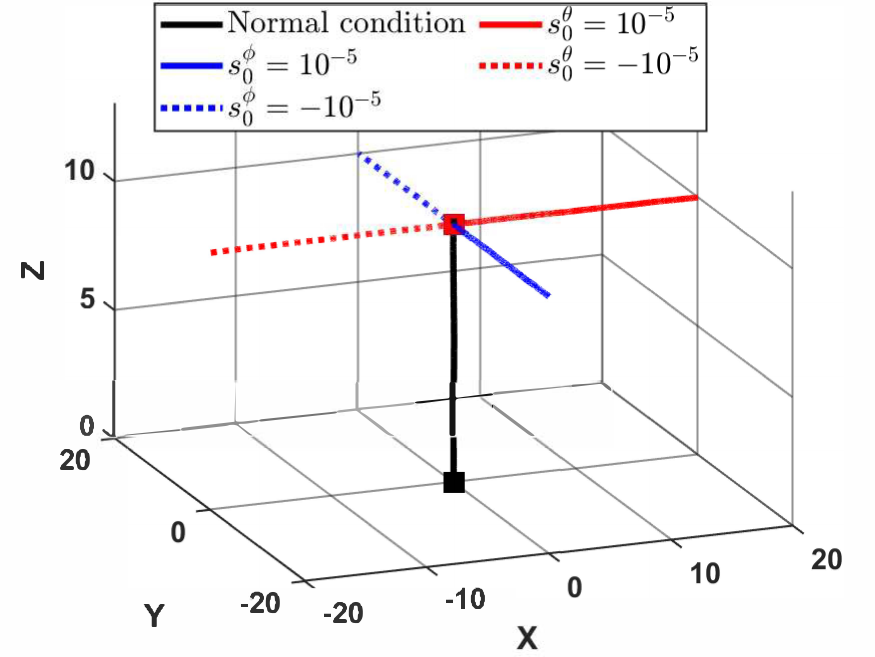}}
\hfill
\subfloat[]{\label{fig:drone_trajectory2} \includegraphics[width=0.223\textwidth, height=3.4cm]{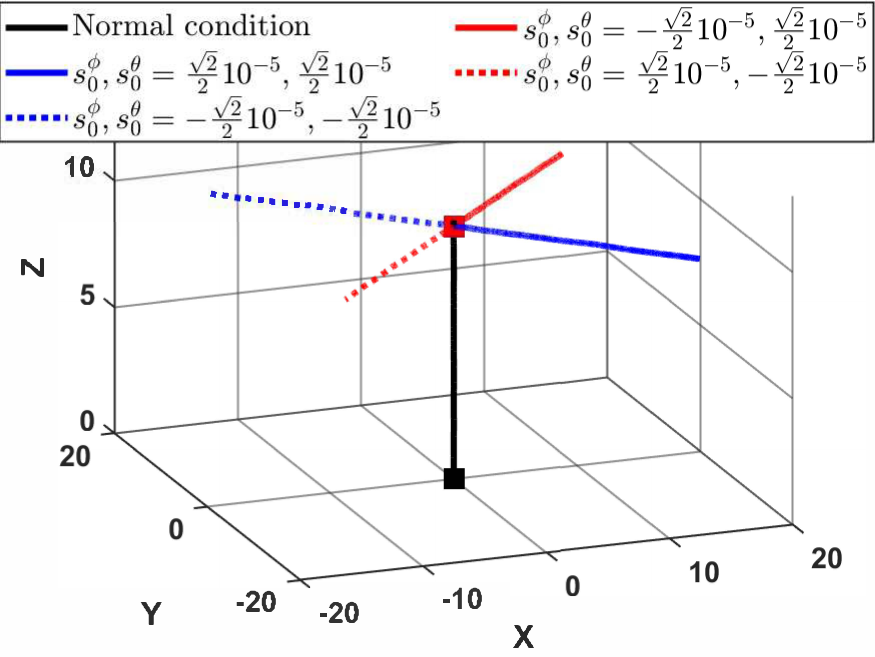}}
\caption{(a,b) The trajectory of the drone for different initial conditions on $s$ with the attack starting at red square. }\label{fig:drone_trajectory}
\end{figure}

%% file: Conclusion.tex
\section{Conclusion}
\label{sec:concl}
In this paper, we focused on vulnerability analysis for nonlinear control systems with Gaussian noise, when attacker can compromise sensor measurements from any subset of sensors. Notions of strict stealthiness and $\epsilon$-stealthiness were defined, and we   showed that these notions are independent of the deployed intrusion detector. 
Using the KL-divergence, 
we presented conditions for the existence of stealthy yet effective attacks. Specifically, we defined the $(\epsilon,\alpha)$-successful attacks where the goal of the attacker is to be $\epsilon$-stealthy while causing deviation in the trajectory of states with respect to system’s own desired unattacked trajectory, determined by the parameter $\alpha$. 
Depending on the level of attacker’s knowledge about the plant model, controller,
and the system states, two different attack models are considered and for each attack model, we then derived a condition for which there exists a sequence of such $(\epsilon,\alpha)$-successful false-data injection  attacks. 
We also provided the results for LTI systems, showing that they are compatible with the existing results for LTI systems and $\chi^2$-based detectors. Finally, we considered the impact of stealthy attacks
on state estimation, and showed that if the the closed-loop control system including the estimator is incrementally stable, then the state estimation in the presence of attack converges to the attack-free estimates.


%% file: appendix.tex
\section*{Appendix}
\textbf{Proof of Lemma~\ref{lemma:determin}:}

\begin{proof}
According to~\eqref{eq:control_rec_d} and~\eqref{eq:control_rec}, $\mathpzc{X}_0^s$ and $y_0^{c,a}:y_{i}^{c,a}$ are sufficient to determine $u_{i}^s$ for any $i\geq 0$ as rewrite $u_{i}^s$~as
\begin{equation}
u_{i}^s=F_{u_{i}^s}(\mathpzc{X}_0^s, y_0^{c,a}:y_{i}^{c,a}),
\end{equation}
where $F_{u_{i}^s}$ is an deterministic mapping. Therefore, given $\mathbf{Z}_{-T}^{-1},\mathbf{Z}^f_{0}:\mathbf{Z}^f_{t}$ we can determine $u_{i}^s$ for any $0\leq i \leq t$. Similarly, using~\eqref{eq:attack_trajec_2} and~\eqref{eq:free_trajec}, $\mathpzc{X}_{-T}$, $y_{-T}^{c}:y_{-1}^{c}$ and $y_0^{c,a}:y_{i}^{c,a}$ are sufficient to determine $u_{i}^a$ for any $i\geq 0$. Alternatively, for some deterministic mapping $F_{u_{i}^a}$ we have
\begin{equation}
u_{i}^a=F_{u_{i}^a}(\mathpzc{X}_{-T},y_{-T}^{c}:y_{-1}^{c},y_0^{c,a}:y_{i}^{c,a}).
\end{equation}
This means if  $\mathbf{Z}_{-T}^{-1},\mathbf{Z}^f_{0}:\mathbf{Z}^f_{t}$ are given, then $u_{i}^a$ will also be directly obtained for any $0\leq i \leq t$. 

Using~\eqref{eq:state_es}, we have $\hat{x}^a_0=\mathcal{E}_0(y_{-\mathcal{T}}^{c}:y_{-1}^{c})$ which means that if we have $y_{-\mathcal{T}}^{c}:y_{-1}^{c}$, then we can determine $\hat{x}_0$. On the other hand, using~\eqref{eq:out_attack} we have 
\begin{equation}
y_t^a=y_t^{c,a}+h(\hat{x}_t^a)-h(\hat{x}_t^a-s_t).
\end{equation}
Thus, given $y_0^{c,a}$ and $\hat{x}_0^a$ and $s_0$, we can determine $y_0^{a}$. Since we already showed $y_{-\mathcal{T}}^{c}:y_{-1}^{c}$ determines $\hat{x}_0^a$, therefore, $y_0^{a}$ will be determined by $\mathbf{Z}_{-T}^{-1},\mathbf{Z}^f_{0}:\mathbf{Z}^f_{t}$ -- this sequence includes $y_{-\mathcal{T}}^{c}:y_{-1}^{c}$.  
Now, according to~\eqref{eq:attack_dynamics}, $s_1$ is determined by $s_0$, $\hat{x}_0$ and $u_0^s$, and since all these signal are determined by $\mathbf{Z}_{-T}^{-1},\mathbf{Z}^f_{0}:\mathbf{Z}^f_{t}$, therefore, $s_1$ is also determined. Now, since $y_0^a$ is determined by this sequence, using $\hat{x}^a_1=\mathcal{E}_1(y_{-\mathcal{T}}^{c}:y_{-1}^{c},y_0^a)$ we can determine $\hat{x}^a_1$. This process can continue by induction until time $t$ for which $y_t^a$, $s_{t+1}$ and $\hat{x}_{t+1}^a$ will be determined by  $\mathbf{Z}_{-T}^{-1},\mathbf{Z}^f_{0}:\mathbf{Z}^f_{t}$.
\end{proof}

\vspace{4pt}\noindent\textbf{Proof of Case 2 in Theorem~\ref{thm:PA_2}:}

We start by introducing the following lemma.

\begin{lemma}\label{lemma:determin_2}
Assume the sequence of $\mathbf{Z}_{-T}^{-1},\mathbf{Z}^f_{0}:\mathbf{Z}^f_{t},v'_{-\mathcal{T}}:v'_t$
are given for any $t\geq 0$ and $s_0$, $\mathpzc{X}_{-T}$ and $\mathpzc{X}_0^s$ are chosen deterministically. Then, the signals $u_{t}^a$, $u_{t}^s$, $\hat{x}_{t}^a$ and $s_{t+1}$ are also uniquely (deterministically) obtained.
\end{lemma}

\begin{proof}
We already proved in Lemma~\ref{lemma:determin} that given $\mathbf{Z}_{-T}^{-1},\mathbf{Z}^f_{0}:\mathbf{Z}^f_{t}$, we can determine $u_{t}^a$, $u_{t}^s$. Now, using~\eqref{eq:state_es}, we have $\hat{x}^a_0=\mathcal{E}'_0(y'_{-\mathcal{T}}:y'_{-1},{y'^a_{0}})$ where $y'_k=h'(x_k)+v'_k$ for any $-\mathcal{T}\leq k\leq -1$ and $y'^a_{0}=h'(x_0)+v'_0$. Given $\mathbf{Z}_{-T}^{-1},\mathbf{Z}^f_{0}$ and $v'_{-\mathcal{T}}:v'_0$, then $y'_{-\mathcal{T}}:y'_{-1},{y'^a_{0}}$ are determined and therefore, $\hat{x}^a_0$ is determined. Using~\eqref{eq:attack_dynamics} and the fact that $s_0$ and $u_0^s$ are deterministic, we can conclude that $s_1$ is also deterministic. For $\hat{x}^a_1=\mathcal{E}'_0(y'_{-\mathcal{T}}:y'_{-1},{y'^a_{0}},{y'^a_{1}})$, we have ${y'^a_{1}}=h'(x_1^a)+v'_1=h'(x_1^f+s_1)+v'_1$. Therefore, given  $s_1$ deterministic and having  $\mathbf{Z}^f_{1}, v'_1$, then $y'_1$ will be deterministic. As a result, $\hat{x}_{1}$ is deterministic. This process can continue until time $t$ where we can conclude $\hat{x}_{t}^a$ and $s_{t+1}$ are deterministic given the sequence in Lemma~\ref{lemma:determin_2}.
\end{proof}

Using the monotonicity property of the KL-divergence from Lemma~\ref{lemma:mon} it holds that
\begin{equation}\label{ineq:5_3}
\begin{split}
K&L\big(\mathbf{Q}(Y_{-T}^{-1},Y_{0}^a:Y_t^a)||\mathbf{P}(Y_{-T}:Y_t)\big)  \\
\leq &  KL\big(\mathbf{Q}(\mathbf{Z}_{-T}^{-1},\mathbf{Z}_{0}^f:\mathbf{Z}_{t}^f,v'_{-\mathcal{T}}:v'_t)||\mathbf{P}(\mathbf{Z}_{-T}:\mathbf{Z}_{t},v'_{-\mathcal{T}}:v'_t)\big).
\end{split}
\end{equation}
Then, we apply the chain-rule property of KL-divergence on the right-hand side of 
\eqref{ineq:5_2} to obtain the following
\begin{equation}\label{ineq:7_3}
\begin{split}
&KL\big(\mathbf{Q}(\mathbf{Z}_{-T}^{-1},\mathbf{Z}_{0}^f:\mathbf{Z}_{t}^f,v'_{-\mathcal{T}}:v'_t)||\mathbf{P}(\mathbf{Z}_{-T}:\mathbf{Z}_{t},v'_{-\mathcal{T}}:v'_t)\big)\\
&= KL\big(\mathbf{Q}(\mathbf{Z}_{-T}^{-1},v'_{-\mathcal{T}}:v'_t)||\mathbf{P}(\mathbf{Z}_{-T}^{-1},v'_{-\mathcal{T}}:v'_t)\big)+\\
&  KL\big(\mathbf{Q}(\mathbf{Z}_{0}^f:\mathbf{Z}_{t}^f|\mathbf{Z}_{-T}^{-1},v'_{-\mathcal{T}}:v'_t)||\mathbf{P}(\mathbf{Z}_{0}:\mathbf{Z}_{t}|\mathbf{Z}_{-T}^{-1},v'_{-\mathcal{T}}:v'_t)\big)=\\
& KL\big(\mathbf{Q}(\mathbf{Z}_{0}^f:\mathbf{Z}_{t}^f|\mathbf{Z}_{-T}^{-1},v'_{-\mathcal{T}}:v'_t)||\mathbf{P}(\mathbf{Z}_{0}:\mathbf{Z}_{t}|\mathbf{Z}_{-T}^{-1},v'_{-\mathcal{T}}:v'_t)\big);
\end{split}
\end{equation}
where we used the property that the KL-divergence of two identical distributions (i.e., $\mathbf{Q}(\mathbf{Z}_{-T}^{-1},v'_{-\mathcal{T}}:v'_t)$ and $\mathbf{P}(\mathbf{Z}_{-T}:\mathbf{Z}_{t},v'_{-\mathcal{T}}:v'_t)$ is zero. Now, we apply the chain-rule property of
KL-divergence to~\eqref{ineq:7_3} and get
\begin{align}\label{ineq:8_3}
& KL\big(\mathbf{Q}(\mathbf{Z}_{0}^f:\mathbf{Z}_{t}^f|\mathbf{Z}_{-T}^{-1},v'_{-\mathcal{T}}:v'_t)||\mathbf{P}(\mathbf{Z}_{0}:\mathbf{Z}_{t}|\mathbf{Z}_{-T}^{-1},v'_{-\mathcal{T}}:v'_{t})\big) \nonumber\\
&= \sum_{k=0}^{t}\Big\{KL\big(\mathbf{Q}(x_k^f|\mathbf{Z}_{-T}^{-1},\mathbf{Z}_{0}^f:\mathbf{Z}_{k-1}^f,v'_{-\mathcal{T}}:v'_t)\nonumber\\
&\qquad  \qquad \qquad  ||\mathbf{P}(x_{k}|\mathbf{Z}_{-T}:\mathbf{Z}_{k-1},v'_{-\mathcal{T}}:v'_{t})\big)\nonumber\\
& \qquad + KL\big(\mathbf{Q}(y_k^{c,a}|x_k^f,\mathbf{Z}_{-T}^{-1},\mathbf{Z}_{0}^f:\mathbf{Z}_{k-1}^f,v'_{-\mathcal{T}}:v'_{t})\nonumber\\
&\qquad  \qquad \qquad  ||\mathbf{P}(y_{k}|x_k,\mathbf{Z}_{-T}:\mathbf{Z}_{k-1},v'_{-\mathcal{T}}:v'_{t})\big)\Big\}.
\end{align}

Given $\mathbf{Z}_{-T}:\mathbf{Z}_{k-1}$ and the fact that $w_{k-1}$ is independent of $v'_{-\mathcal{T}}:v'_k$, the distribution of $x_k$ is a Gaussian with mean $f(x_{k-1},u_{k-1})$ and covariance $\mathbf{R}_w$. On the other hand, using equations~\eqref{eq:error_a} and~\eqref{eq:x_f_def} we have 
\begin{equation}\label{eq:fake_state_traj_2}
\begin{split}
x_k^f=&f(x_{k-1}^f+s_{k-1},u_{k-1}^a)-f(\hat{x}_{k-1}^a,u_{k-1}^s)\\
&+f(\hat{x}_{k-1}^a-s_{k-1},u_{k-1}^s)+w_{k-1}
\end{split}
\end{equation}

Using Lemma~\ref{lemma:determin_2}, given  $\mathbf{Z}_{-T}^{-1},\mathbf{Z}_{0}^f:\mathbf{Z}_{k-1}^f,v'_{-\mathcal{T}}:v'_{t}$ and the fact that $s_0$, $\mathpzc{X}_{-T}$ and $\mathpzc{X}_0^s$ are deterministic, $\hat{x}_{k-1}^a$, $s_{k-1}$, $u_{k-1}^s$ and $u_{k-1}^a$ are also determined. Therefore, the distribution of $x_k^f$ given $\mathbf{Z}_{-T}^{-1},\mathbf{Z}_{0}^f:\mathbf{Z}_{k-1}^f,v'_{-\mathcal{T}}:v'_{t}$ is a Gaussian with mean $f(x_{k-1}^f+s_{k-1},u_{k-1}^a)-f(\hat{x}_{k-1}^a,u_{k-1}^s)+f(\hat{x}_{k-1}^a-s_{k-1},u_{k-1}^s)$ and covariance $\mathbf{R}_w$. Using~\eqref{eq:free_trajec2} and~\eqref{eq:fake_state_traj_2}, it holds that 
\begin{equation}\label{eq:mean_diff_2}
\begin{split}
x_k^f-x_k=&f(x_{k-1}^f+s_{k-1},u_{k-1}^a)-f(\hat{x}_{k-1}^a,u_{k-1}^s)\\
&+f(\hat{x}_{k-1}^a-s_{k-1},u_{k-1}^s)-f(x_{k-1},u_{k-1}),
\end{split}
\end{equation}
for all $0\leq k \leq t$. Now, according to Lemmas~\ref{lemma:chain}, ~\ref{lemma:Guassian}, ~\ref{lemma:maximum} and equation~\eqref{eq:mean_diff}, for all $0\leq k \leq t$ we have
\begin{align}\label{ineq:10}
&KL\big(\mathbf{Q}(x_k^f|\mathbf{Z}_{-T}:\mathbf{Z}_{k-1}^f,v'_{-\mathcal{T}}:v'_{t})||\mathbf{P}(x_{k}|\mathbf{Z}_{-T}:\mathbf{Z}_{k-1},v'_{-\mathcal{T}}:v'_{t})\big) \nonumber\\
&= \mathbb{E}_{\mathbf{Q}(x_k^f,\mathbf{Z}_{-T}^{-1},\mathbf{Z}_{0}^f:\mathbf{Z}_{k-1}^f,v'_{-\mathcal{T}}:v'_{t})}\bigg\{(x_k-x_k^f)^T \mathbf{R}_w^{-1}  (x_k-x_k^f)\bigg\}\nonumber\\
&\leq \lambda_{max}(\mathbf{R}_w^{-1})\Vert x_k-x_k^f \Vert^2.
\end{align}
Moreover, given $x_k$, the distribution of $y_k$ is a Gaussian with mean $h(x_k)$ and covariance $\mathbf{R}$.Using Lemma~\ref{lemma:determin_2}, given $x_k^f,\mathbf{Z}_{-T}^{-1},\mathbf{Z}_{0}^f:\mathbf{Z}_{k-1}^f,v'_{-\mathcal{T}}:v'_{k}$ then $y_k$ is Gaussian with mean $h(x_k^f)+\sigma'_k$ and covariance $\mathbf{R}$. Therefore, combining Lemmas~\ref{lemma:chain}, ~\ref{lemma:Guassian}, ~\ref{lemma:maximum} and equation~\eqref{eq:mean_diff}, for all $0\leq k \leq t$ we get
\begin{align}\label{ineq:12}
&KL\big(\mathbf{Q}(y_k^{c,a}|x_k^f,\mathbf{Z}_{-T}^{-1},\mathbf{Z}_{0}^f:\mathbf{Z}_{k-1}^f,v'_{-\mathcal{T}}:v'_{k})||\mathbf{P}(y_{k}|x_k)\big)\nonumber\\
&~~ =\mathbb{E}_{\mathbf{Q}(x_k^f,\mathbf{Z}_{-T}^{-1},\mathbf{Z}_{0}^f:\mathbf{Z}_{k-1}^f,v'_{-\mathcal{T}}:v'_{k})}\Big\{\big(h(x_k^f)+\sigma'_k-h(x_k)\big)^T\nonumber\\
& \quad \times  \mathbf{R}_v^{-1}  \big(h(x_k^f)+\sigma'_t-h(x_k)\big)\Big\} \leq  \lambda_{max}(\mathbf{R}_v^{-1})\nonumber \\
& \quad \times \Big(L_h^2\Vert x_k-x_k^f \Vert^2+2L_h^2b'_{\zeta}\Vert x_k-x_k^f \Vert+L_h^2{b'_{\zeta}}^2\Big).
\end{align}

As in Case~1, we combine~\eqref{ineq:5_3}-\eqref{ineq:8_3},~\eqref{ineq:12} and~\eqref{ineq:10} to~get 
\begin{equation}\label{ineq:final_2_2}
\begin{split}
K&L\big(\mathbf{Q}(Y_{-T}^{-1},Y_{0}^a:Y_t^a)||\mathbf{P}(Y_{-T}:Y_t)\big) \leq\\
 &  \frac{\kappa^2 \Vert s_0\Vert^2 }{1-\lambda^2} \big(\lambda_{max}(\mathbf{R}_w^{-1}) + L_h^2\lambda_{max}(\mathbf{R}_v^{-1})\big) \\
&+ \frac{\kappa \Vert s_0\Vert }{1-\lambda} \big(2\gamma(0)\lambda_{max}(\mathbf{R}_w^{-1}) + L_h^2\lambda_{max}(\mathbf{R}_v^{-1})(2\gamma(0)+b'_{\zeta})\big) \\
&+ \sum_{k=0}^{t} \gamma^2(k) \big(\lambda_{max}(\mathbf{R}_w^{-1}) + L_h^2\lambda_{max}(\mathbf{R}_v^{-1})\big)\\
&+\sum_{k=0}^{t} \gamma(k) \big(2L_h^2b'_{\zeta}\lambda_{max}(\mathbf{R}_v^{-1})\big)+\sum_{k=0}^{t} {b'_{\zeta}}^2L_h^2\lambda_{max}(\mathbf{R}_v^{-1}).
\end{split}
\end{equation}
Again, similar to Case~1, if we denote the right-hand side of the above inequality by $b_{\epsilon}$ and use Theorem~\ref{thm:stealthy}, the system~$\Sigma_{\text{\rom{2}}}$ will be $\epsilon$-stealthy attackable with $\epsilon=\sqrt{1-e^{-b_{\epsilon}}}$.

\subsection{Discussion on Incremental Stability}
\label{app:IS}
As previously discussed, the notions of incremental stability and contraction theory are well-defined 
\cite{angeli2002lyapunov,tran2018convergence,tran2016incremental} and going over their satisfying  conditions is beyond the scope of this paper. However, to provide some useful intuition, we start with the following example. 

\begin{example}
Consider the system 
\begin{equation*}
x_{t+1}=x_{t}+a x_{t}^3 +bu_t+w_t,
\end{equation*}
where $a>0$ and $w$ is a Gaussian noise with zero mean and variance $\mathbf{R}_w$. Also, consider the feedback controller 
$$u_t=-\frac{a}{b}x_t^3-\frac{1}{2b}x_t.$$ 
With such control input, the closed-loop system becomes 
\begin{equation}\label{eq:example1}
x_{t+1}=0.5x_{t}+w_t.  
\end{equation}
The above system is 
GIES since the difference between any two trajectories 
of~\eqref{eq:example1} with the same noise $w_0:w_{t-1}$ and different initial conditions ($\xi_1$ and $\xi_2$), defined as $\Delta x_{t}=x(t,\xi_1,w)-x(t,\xi_2,w)$, has the dynamics $\Delta x_{t+1}=.5\Delta x_{t}$ with the nonzero initial condition $\Delta x_{0}=\xi_1-\xi_2$. Hence, $\Delta x_t$  converge exponentially to zero and the closed-loop control system is exponentially stable. 

Now, let us consider
\begin{equation}\label{eq:example2}
x_{t+1}=x_{t}+a x_{t}^3 +d_t,    
\end{equation}
with $d_t=bu_t+w_t$. The difference between two trajectories 
with the same input $\{d_0,...,d_{t-1}\}$ and different initial condition ($\xi'_1$ and $\xi'_2$) is $\Delta x'_{t}=x(t,\xi'_1,d)-x(t,\xi'_2,d)$.  Therefore, the difference dynamics becomes $\Delta x'_{t+1}=\Delta x'_{t}+a (x^3(t,\xi'_1,d)-x^3(t,\xi'_2,d))$. One can use Lyapunov function method by defining $V_t={\Delta x_{t}'}^2$ and verify that
\begin{equation*}
\begin{split}
V_{t+1}&-V_t= {\Delta {x'}^2_{t+1}}-{\Delta x'_{t}}^2  = 2a{\Delta x'_{t}}^2\Big(x^2(t,\xi'_1,d)\\
&+x^2(t,\xi'_2,d)+x(t,\xi'_1,d)x(t,\xi'_2,d)\Big)+a^2{\Delta x'_{t}}^2\times \\
&\Big(x^2(t,\xi'_1,d)+x^2(t,\xi'_2,d)
+x(t,\xi'_1,d)x(t,\xi'_2,d)\Big).
\end{split}
\end{equation*}
Since $x^2(t,\xi'_1,d)+x^2(t,\xi'_2,d)+x(t,\xi'_1,d)x(t,\xi'_2,d)>0$ and $V_0>0$, it holds that $V_{t+1}-V_t>0$ for all $t\in \mathbb{Z}_{\geq 0}$; thus, 
$\Delta x'_{t}$ diverges 
and 
\eqref{eq:example2} is globally incrementally unstable. 

\end{example}